\documentclass[journal, 10pt]{IEEEtran}

\usepackage{cite}

\ifCLASSINFOpdf
  \usepackage[pdftex]{graphicx}
	\graphicspath{{./}}
\else
\fi

\usepackage[cmex10]{amsmath}
\usepackage{bbm}
\usepackage{blkarray}
\usepackage{algorithmic}
\usepackage[boxed]{algorithm2e}
\usepackage{amssymb}
\usepackage{amsthm}
\usepackage[caption=false, font=footnotesize]{subfig}
\usepackage{float}
\usepackage{url}
\usepackage{xcolor}

\usepackage{multirow}
\newtheorem{prop}{Proposition}
\newtheorem{lemm}{Lemma}

\usepackage{enumitem}

\makeatletter
\long\def\@makecaption#1#2{\ifx\@captype\@IEEEtablestring%
\footnotesize\begin{center}{\normalfont\footnotesize #1}\\
{\normalfont\footnotesize\scshape #2}\end{center}%
\@IEEEtablecaptionsepspace
\else
\@IEEEfigurecaptionsepspace
\setbox\@tempboxa\hbox{\normalfont\footnotesize {#1.}~~ #2}%
\ifdim \wd\@tempboxa >\hsize%
\setbox\@tempboxa\hbox{\normalfont\footnotesize {#1.}~~ }%
\parbox[t]{\hsize}{\normalfont\footnotesize \noindent\unhbox\@tempboxa#2}%
\else
\hbox to\hsize{\normalfont\footnotesize\hfil\box\@tempboxa\hfil}\fi\fi}
\makeatother

\usepackage{comment}
\begin{document}

\title{LoRa beyond ALOHA: An Investigation of Alternative Random Access Protocols}

\author{Luca~Beltramelli,
        Aamir~Mahmood,
				Patrik~{\"O}sterberg,
				and~Mikael Gidlund
\thanks{The authors are with the Department of Information Systems and Technology, Mid Sweden University, Holmgatan 10, 850~71~Sundsvall, Sweden. Corresponding author: Luca Beltramelli, email: luca.beltramelli@miun.se}
\vspace{-20pt}}

\maketitle

\begin{abstract}

We present a stochastic geometry-based model to investigate alternative medium access choices for LoRaWAN---a widely adopted low-power wide-area networking (LPWAN) technology for the Internet-of-things (IoT). LoRaWAN adoption is driven by its simplified network architecture, air interface, and medium access. The physical layer, known as LoRa, provides quasi-orthogonal virtual channels through spreading factors (SFs) and time-power capture gains. However, the adopted pure ALOHA access mechanism suffers, in terms of scalability, under the same-channel same-SF transmissions from a large number of devices. In this paper, our objective is to explore access mechanisms beyond-ALOHA for LoRaWAN. Using recent results on time- and power-capture effects of LoRa, we develop a unified model for the comparative study of other choices, i.e., slotted ALOHA and carrier-sense multiple access (CSMA). The model includes the necessary design parameters of these access mechanisms, such as guard time and synchronization accuracy for slotted ALOHA, carrier sensing threshold for CSMA. It also accounts for the spatial interaction of devices in annular shaped regions, characteristic of LoRa, for CSMA. The performance derived from the model in terms of coverage probability, channel throughput, and energy efficiency are validated using Monte-Carlo simulations. Our analysis shows that slotted ALOHA indeed has higher reliability than pure ALOHA but at the cost of lower energy efficiency for low device densities. Whereas, CSMA outperforms slotted ALOHA at smaller SFs in terms of reliability and energy efficiency, with its performance degrading to pure ALOHA at higher SFs.
\end{abstract}
\IEEEpeerreviewmaketitle

\section{Introduction}
\label{section_intro}

\IEEEPARstart{L}{}ow-power wide-area network (LPWAN) technologies are rapidly being adopted in delay-tolerant industrial applications for wide-area sensing, monitoring, and supervisory control with the emphasis on energy-efficient support for a massive number of devices over long distances~\cite{Candell}.
{The energy efficiency of
LPWANs make them a cost-effective candidate solution for non-critical industrial applications~\cite{sanchez2016state}. The robust modulations that allow LPWANs to achieve wide-area coverage can be useful in industrial environments where the wireless channel is often affected by multipath and fading~\cite{luvisotto2018use}. Moreover, the extended coverage reduces the need for multi-hop communication while supporting mobility.} Among others, the most critical challenge that LPWANs face is the interference caused by simultaneous transmissions from a large number of connected devices. To manage and combat interference, the selection of effective and yet energy efficient multiple access mechanisms is paramount. It is, hence, essential to explore all potential medium access design choices available for LPWANs to best support all use cases. This work investigates possible medium access solutions for LoRa, an LPWAN technology that has recently become popular both in industry and academia.

 LoRa, as a communication technology, is the union of a proprietary modulation scheme at the physical (PHY) layer with LoRaWAN, the latter being an open standard defining the medium access and other higher layers of the communication stack \cite{centenaro2016long}. According to LoRaWAN specifications, most devices access the channel using a mechanism similar to pure ALOHA. ALOHA is a well-understood protocol known for its simplicity, ease of implementation, and for adding a minimal communication overhead. The random access nature of ALOHA perfectly matches to the sporadic data transmissions of typical monitoring applications. However, ALOHA is known to suffer from limited scalability (i.e. the maximum number of devices that can be served) and the lack of any QoS guarantees. Numerous studies \cite{georgiou2017low, mahmood2019scalability, liando2019known, croce2018impact, rahmadhani2018lorawan, haxhibeqiri2017lora}, investigating the performance of LoRa, have highlighted its qualities but also its limits, some, in particular, attributable to the ALOHA protocol. Many works in literature address this latter aspect by proposing alternative medium access mechanisms in LoRaWAN \cite{polonelli2019slotted,pham2018robust,to2018simulation,leonardi2018industrial,piyare2018demand,haxhibeqiri2018low,zorbas2020ts,rizzi2017using}. The two main research directions of these studies are,  a)  introducing QoS support in LoRa via the use of scheduling mechanisms for deterministic data transmissions in time-critical industrial use-cases and b) enhancing LoRa scalability via the use of random access mechanisms such as slotted ALOHA and CSMA.  
This work however follows the second direction, as improving the reliability (i.e. transmission success probability) and scalability of LoRa without compromising its energy efficiency would not only benefit the already supported industrial applications but also enable the use of LoRa in new applications (e.g. sensor backhaul, machine health monitoring)~\cite{Candell}.

To the best of the authors’ knowledge, the performance of alternative random access schemes for LoRa has not been investigated nor compared analytically in the literature. Existing works (e.g., \cite{georgiou2017low}, \cite{mahmood2019scalability}) proposing analytical models for LoRa have exclusively analyzed the behavior of pure ALOHA. Seminal works in stochastic geometry \cite{haenggi2009interference,baccelli2010stochastic} have compared the performance of different access mechanisms, but their models  do not capture the PHY layer characteristics of LoRa. In this paper, we develop an analytical model for performance analysis of LoRa networks under three basic random access mechanisms: pure ALOHA (\mbox{P-ALOHA}), slotted ALOHA (\mbox{S-ALOHA}) and non-persistent carrier-sense multiple access (\mbox{NP-CSMA}). The objective is to analyze how these access mechanisms can improve the performance of a LoRa network, especially channel throughput and energy efficiency while taking into account the LoRa modulation and interference specific details. The proposed model can be a helpful tool in the design and evaluation of random access mechanisms for LoRa networks.  
In comparison with the existing works proposing analytical models for the performance analysis of LoRa, the novel contributions of this paper are as follows:\begin{itemize}
 \item A unified analytical model for investigating and comparing the performances of three random access mechanisms \mbox{P-ALOHA}, \mbox{S-ALOHA}, and \mbox{NP-CSMA} in conjunction with LoRa chirp spread spectrum (CSS) modulation. Based on the latest reported results in the literature, the model incorporates the time- and power-capture effects and interference behavior of LoRa CSS technology. 
 \item Characterization of the interference intensity in terms of the main parameters of the studied random access mechanisms,  qualifies the model for the design of media access control in LoRa: the guard time in \mbox{S-ALOHA} or the sensing threshold in \mbox{CSMA}. 
\item A comparison of the studied access mechanisms in terms of success probability, channel throughput, and energy efficiency. The resulting analysis can be used to explore the best random access mechanism within the LoRa parameter space (Spreading factor, number of devices, deployment area).
\end{itemize}

The rest of this paper is organized as follows. Sec.~\ref{sec:bckg} gives background information on LoRa and LoRaWAN and discusses the related works on random access mechanisms for LoRa. Sec.~\ref{se:Mat_Model} describes the proposed analytical approach to model different access schemes in LoRa.  The analytical and simulation results are presented in Sec.~\ref{sec:Results_Analysis}. Finally, Sec.~\ref{sec:conclusions} concludes this paper.

\section{Background and Related Works}
\label{sec:bckg}
This section introduces LoRa and presents the related works that have proposed alternatives to the default channel access used by LoRaWAN.
\subsection{LoRa and LoRaWAN}
\label{sec:LoRa}
LoRa is a chirp spread spectrum (CSS)-based proprietary modulation and coding scheme developed by Semtech for sub-1 GHz ISM bands. LoRa supports six quasi-orthogonal spreading factors (SF), where the use of a larger SF increases the coverage range (by lowering the receiver sensitivity) at the cost of datarate \cite{centenaro2016long}. The LoRaWAN open standard  specifies the network entities and their roles, and how devices access the shared channel. In a LoRaWAN network, the two main entities related to the air interface are gateways and end devices (EDs).  LoRaWAN specifies three classes of EDs based on their downlink response time and energy consumption. Class A devices offer the best energy-saving performance by waking up only when they have data to transmit using ALOHA. Class B devices wake up at periodic intervals to synchronize and exchange data with the gateway, which makes them a candidate for S-ALOHA based access. Class C devices are always active, continuously listening to the channel for transmissions from the gateway.

\subsection{Alternative channel access mechanisms for LoRa.}
Several works investigating the use of LoRa in industrial scenarios for non time-critical monitoring applications have shown promising results \cite{haxhibeqiri2017lora2, luvisotto2018use}.
Concurrently, new proposals have emerged in recent years to replace \textit{default}  pure ALOHA in LoRaWAN with alternative access mechanisms more suited for industrial applications. The proposals are in general motivated with one of two goals. The first is to improve scalability and communication reliability, which in turn increases energy efficiency.
Any access mechanism that reduces interference in a LoRa network compared to P-ALOHA, can improve the capacity of a gateway to serve more devices. This is especially important for large-scale indoor and outdoor monitoring applications in industrial environments.  To this end, the access techniques proposed in the literature are variants of S-ALOHA, CSMA, or other random access solutions offering better channel utilization than P-ALOHA.
The second motivating factor in upgrading LoRaWAN channel access is to support use cases with relaxed but bounded QoS requirements such as industrial process control applications. This has led to resource reservation-based access solutions based on the allocation of frequency, time, and SF.  Below, we summarize the LoRa-specific channel access enhancements proposed in the literature.

\subsubsection{\mbox{S-ALOHA}}
In theory, by replacing pure with a slotted variant of ALOHA, the channel capacity is doubled by virtue of interference reduction in the network. However, providing energy-efficient synchronization of EDs over the wide coverage of LoRa can be challenging. An S-ALOHA protocol on top of LoRaWAN stack is proposed and implemented in \cite{polonelli2019slotted}.

\subsubsection{CSMA}
In CSMA/CA, EDs sense the channel before attempting a transmission, which reduces the interference in a densely deployed network. Unlike \mbox{S-ALOHA}, an asynchronous CSMA does not require synchronization of EDs. However, its performance is negatively affected by the presence of hidden nodes. In extreme situations where all nodes are hidden, CSMA throughput degrades to \mbox{P-ALOHA}. Since LoRa coverage is several kilometers, a significant percentage of EDs in the service area can be hidden from each other. For sensing, LoRa chipsets support a channel activity detection (CAD) mechanism, designed to detect the presence of LoRa preamble or data symbols on the channel. The CAD mechanism for enabling CSMA in LoRa was experimentally evaluated in \cite{pham2018robust, liando2019known}, while an ns-3 module to simulate p-CSMA in a LoRa network is presented in \cite{to2018simulation}.

\subsubsection{Scheduled MAC}
The use of scheduled MAC in LoRa have been suggested for timely communication of a limited number of devices in industrial applications.
In a scheduled MAC, the channel resource are divided among the EDs to allow for a contention-free transmission. Although attractive for reducing interference, it brings new challenges of transmitting the scheduling information to EDs and predict their resource requirement.  For LoRa, the scheduled MAC techniques proposed in the literature include: hybrid ALOHA/TDMA access for periodic real-time traffic~\cite{leonardi2018industrial}, on-demand TDMA access using wake-up radios~\cite{piyare2018demand}, algorithms to assign timeslots according to EDs' traffic periodicity~\cite{haxhibeqiri2018low}, self-organizing time-slotted communication~\cite{zorbas2020ts}, and time-slotted channel hopping~\cite{rizzi2017using}.

\begin{figure*}[!ht]
     \subfloat[][\centering co-SF interference, $I_1$ $=\text{SF}_{10}$,
     
     $I_2$ $=\text{SF}_{10}$, offset of $0$ symbol time, $\textrm{SIR}=5$ dB \label{subfig-1:interference1}]{%
       \includegraphics[width=0.24\textwidth]{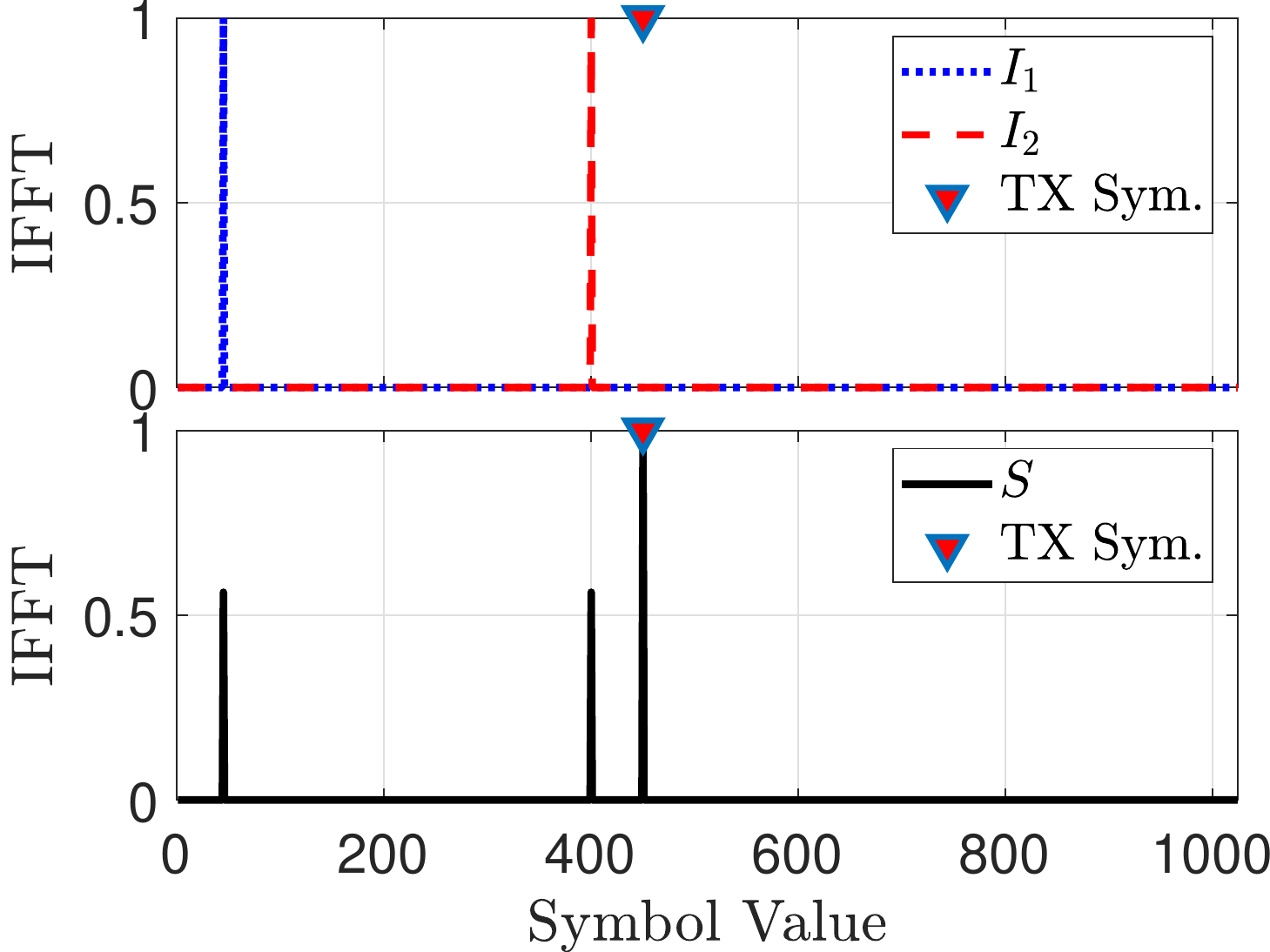}
     }
     \hfill
     \subfloat[][\centering co-SF interference, $I_1$ $=\text{SF}_{10}$,
     
     $I_2$ $=\text{SF}_{10}$, offset of $0.25$ symbol time, $\textrm{SIR}=0$ dB \label{subfig-2:interference2}]{%
       \includegraphics[trim=28 0 0 0,clip,width=0.2235\textwidth]{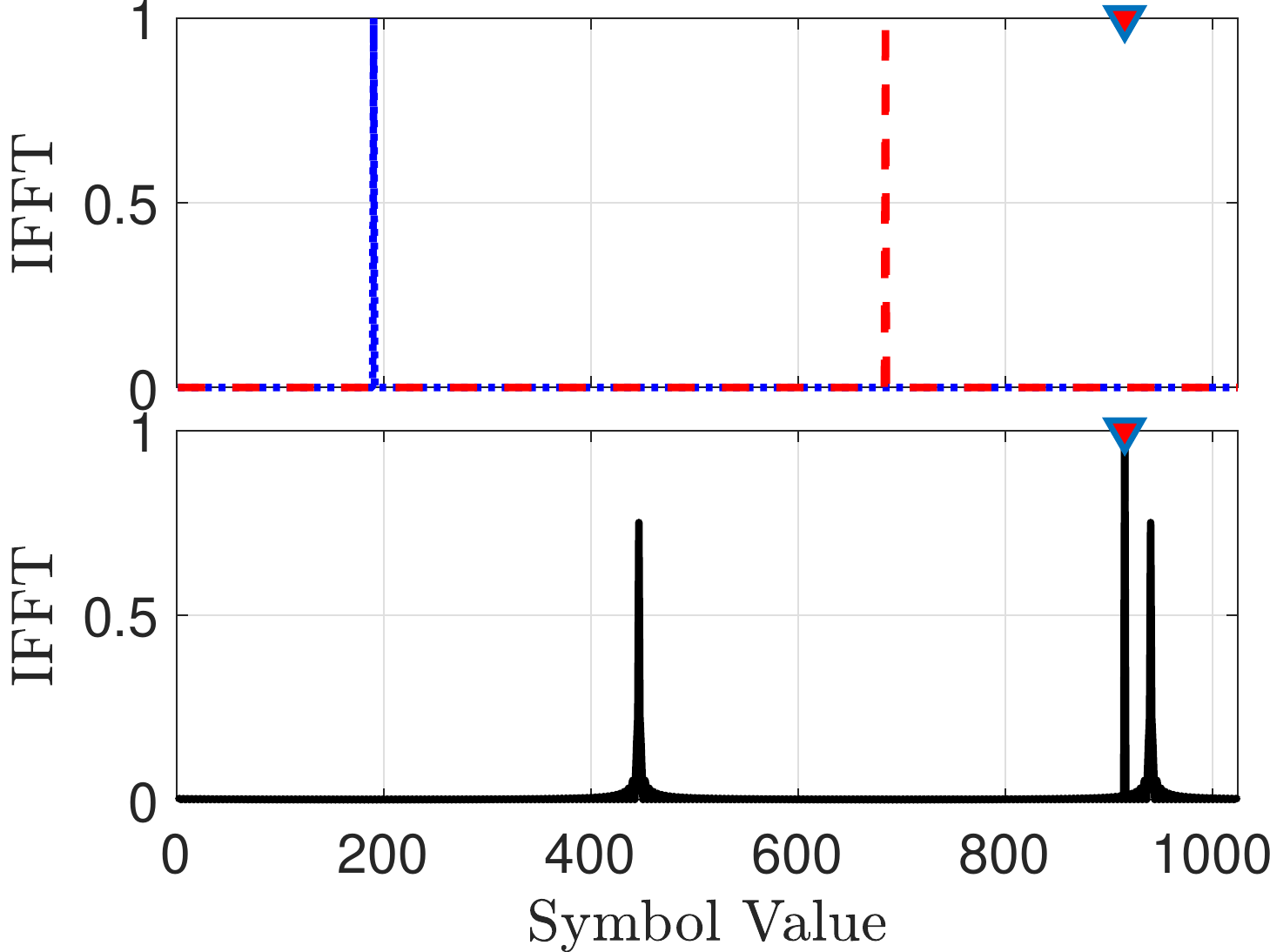}
     }
     \hfill
     \subfloat[][\centering inter-SF interference, $I_1$ $=\text{SF}_{9}$,
     
     $I_2$ $=\text{SF}_{11}$, offset of $0$ symbol time, $\textrm{SIR}=-15$ dB  \label{subfig-3:interference3}]{%
       \includegraphics[trim=28 0 0 0,clip,width=0.2235\textwidth]{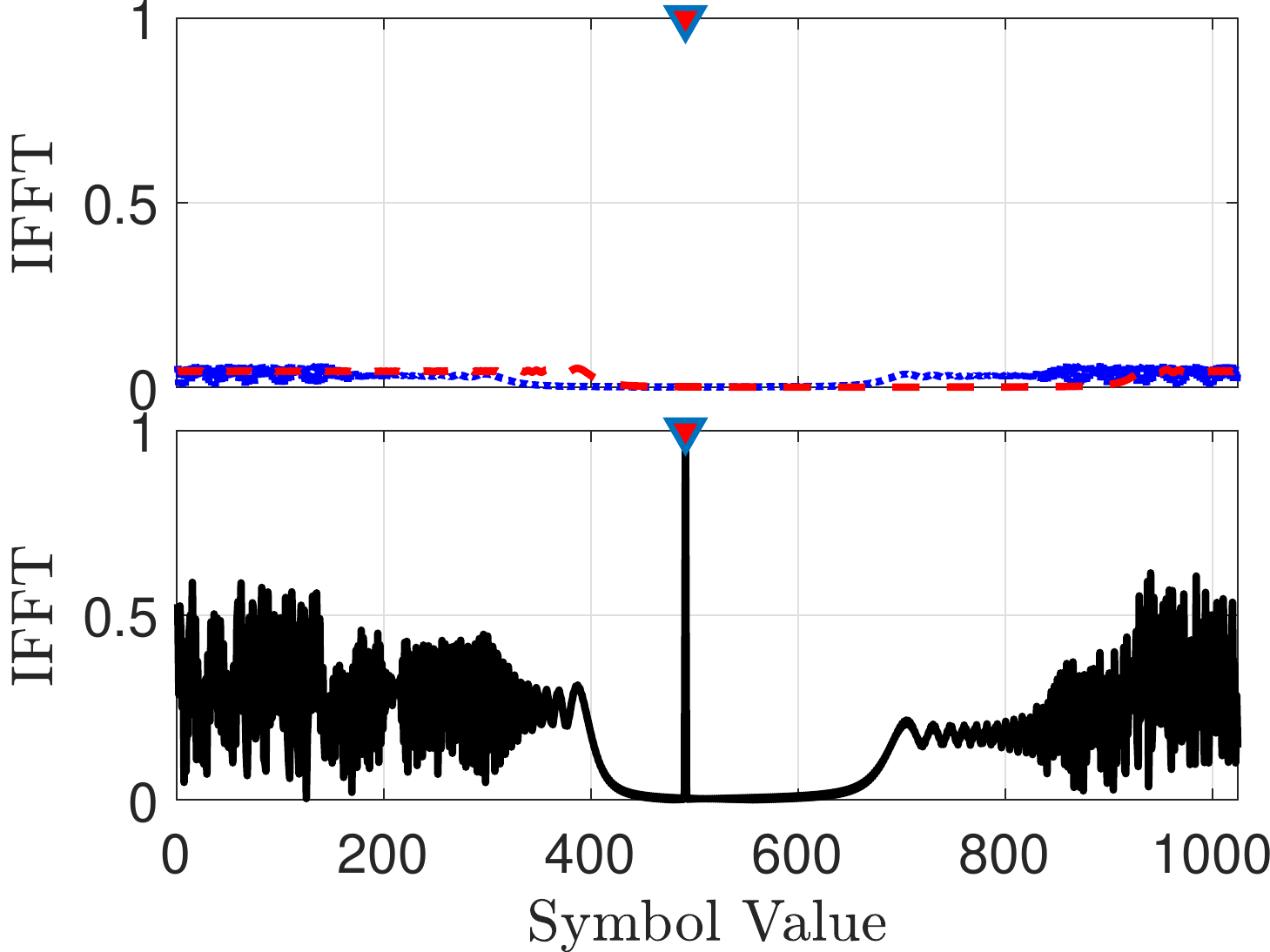}
     }
     \hfill
     \subfloat[][\label{subfig-4:Prsame_sym}]{%
       \includegraphics[width=0.24\textwidth]{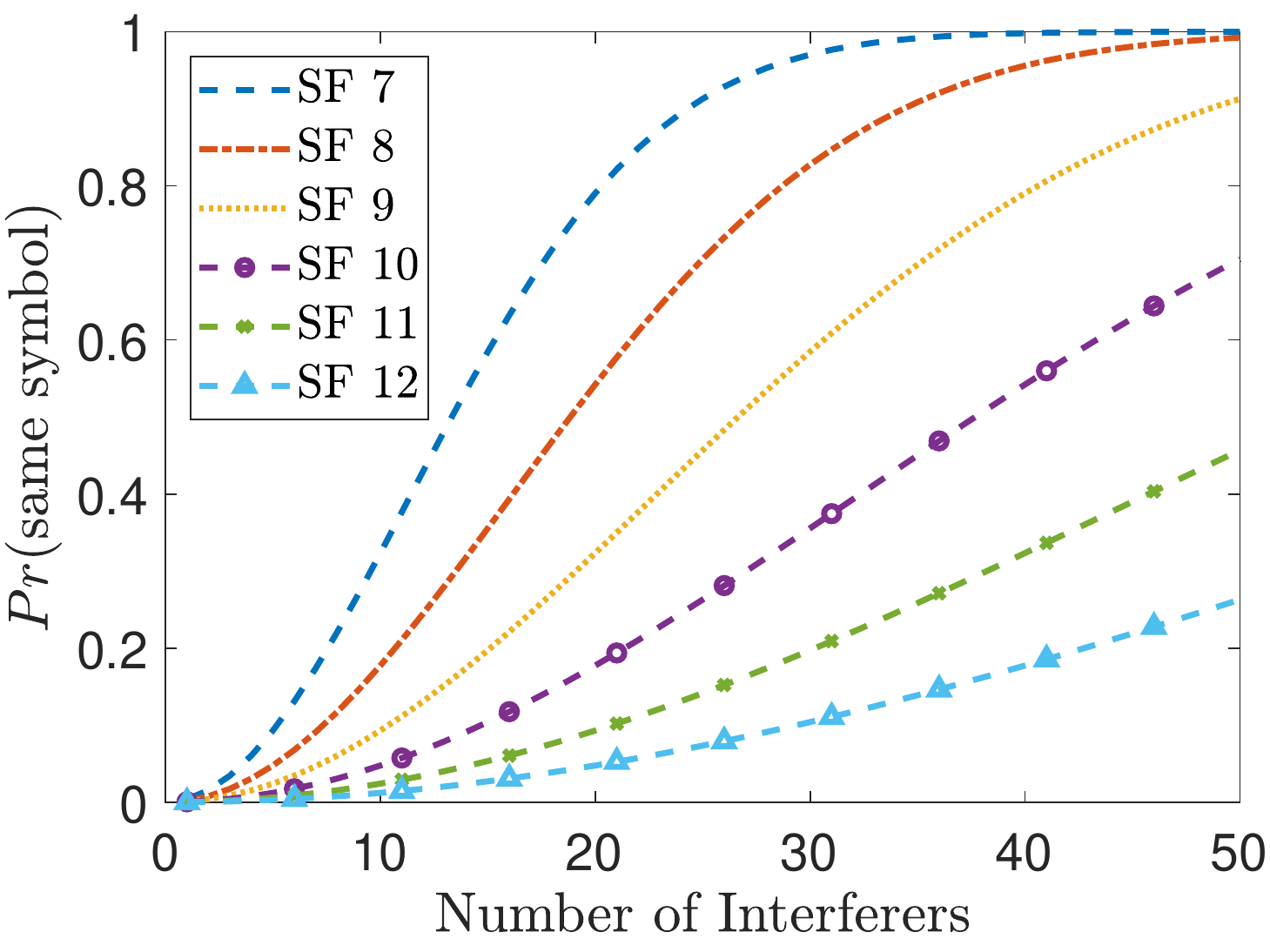}
     }
     \caption{(a-c) IFFT of a signal (S) transmitted using SF~$10$ and affected by two interferers ($I_1$ and $I_2$); (d) Probability that two or more interferers transmitting the same symbol.}
     \label{fig:interference}
     \vspace{-15pt}
\end{figure*}

\section{Mathematical Modeling}
\label{se:Mat_Model}

We consider a LoRa network consisting of a single gateway to provide connectivity to EDs. The EDs are distributed uniformly around the gateway in a service area of radius $R$, according to a homogeneous Poisson Point Process (PPP) $\phi_\text{ED}$ of intensity $\lambda$. The average number of EDs in the area is $\bar{N}=\lambda \pi  R^2$, and the distance distribution of EDs from the gateway is $f_\text{ED}(r)={2r}/{R^2}, 0\leq r\leq R$.  We assume all EDs use the same channel of bandwidth $B$, while each ED selects an SF from the set $\textrm{SF}\in\{7, 8, \cdots, 12\}$. The area is divided into six contiguous annuli, each populated by EDs using the same SF. We consider only unacknowledged uplink traffic (with fixed message size) from all EDs to the gateway, which is prevalent in most LoRa use-cases. 
LoRa devices, if not using listen before talk (LBT), follow regional regulations on maximum duty-cycle  $\alpha$, as a coexistence mechanism in the ISM band. For a fair comparison, we enforce EDs to comply with the same duty-cycle limitations irrespective of their medium access mechanism. We assume there is no restriction on the maximum number of frames that can be simultaneously demodulated at the gateway. In reality, a restriction is imposed by  the number of demodulating paths at the gateway. To account for this hardware limitation, the proposed model could easily be modified as in \cite{sorensen2019analysis}.
In the following section, before developing the proposed model, we establish a fresh case of using dominant co-SF interference based on recent results in the literature.  Following that, in Sec.~\ref{sec:IPPP}, we identify a distinctive parameter (intensity of interferences) in analyzing the performance of different access mechanisms, and analytically derive it for a LoRa network for \mbox{P-ALOHA}, \mbox{S-ALOHA} and \mbox{NP-CSMA}. Later, the success probability conditions are developed in Sec.~\ref{sec:SNR&SIR}, which are then used to find coverage probability, channel throughput, and energy efficiency in Sec.~\ref{sec:perf_metrics}.
\subsection{Interference model}
\label{sec:int_model}
The role of a multiple access mechanism is to dictate how devices access the shared transmission medium. Consequently, the use of different multiple access mechanisms directly leads to different intensities and distribution of interference. 
To correctly predict the performance of different medium access mechanisms, we first discuss the effect of interference in LoRa.

In LoRa CSS modulation,  each symbol is spreaded over a sinusoidal signal with a linearly increasing frequency called an up-chirp. A symbol is represented using a shift of the linear frequency increase.
At the receiver, the signal is multiplied by a second sinusoid with linearly decreasing frequency called down-chirp. The resulting signal produces a sharp peak in the frequency domain corresponding to the symbol value encoded in the chirp.
We analyzed the possible interference scenarios using the simulation tool developed in \cite{croce2018impact}. We identified three notable cases: a) fully overlapping co-SF interference; b) partial overlapping co-SF interference; c) inter-SF interference.
We considered two interfering signals ($I_1$ and $I_2$), and we changed the SF and offset between the useful signal and the interfering signals. 

The results showed in Fig.~\ref{fig:interference} are obtained by considering an FFT-based demodulation, perfectly synchronized to the received signal.
Fig.~\ref{subfig-1:interference1} shows that, if the symbols encoded by the useful and interfering signals are different, the IFFT of fully overlapping co-SF interferers contains a peak at each encoded symbol. The peak value of IFFT depends on the relative received signal strength (RSS) of the signals. Consequently, unless $I_1$ and $I_2$ transmit the same symbols, the receiver can decode the transmitted symbol correctly given that its RSS is higher than the interferers. The signaling alphabet of LoRa has cardinality  $2^\text{SF}$. The probability that from a set of $n$ interfering signals with the same SF, two or more  transmit the same symbol is 
\begin{equation}
 \Pr(\text{Same symbol})\approx 1-e^{-\frac{n^2}{2^\text{SF+1}}}\text{.}
 \label{eq:Pr_samesymbol}
\end{equation}
Fig.~\ref{subfig-4:Prsame_sym} shows \eqref{eq:Pr_samesymbol} for different SFs and number of interferers. Fig.~\ref{subfig-2:interference2} illustrates the case of partially overlapping co-SF interference with the transmitted symbol. The peaks in the IFFT caused by the interfering signals are lower and wider for even signal-to-interference ratio ($\textrm{SIR}$)$ = 0$ dB. Fig.~\ref{subfig-3:interference3} shows the inter-SF interference scenario, where the signal is transmitted using SF~10, and the interferers are using SF~9 and SF~11. 
A large number of interferers with high signal strength are required before the receiver fails to detect the peak of the transmitted symbol.

The ability to receive partially overlapping LoRa transmissions by a LoRa receiver is evaluated in~\cite{rahmadhani2018lorawan} and~\cite{haxhibeqiri2017lora}. These studies show that a) LoRa does not offer any time capture, unlike other spread spectrum techniques, and b) the message can be correctly decoded even if the first few preamble symbols are overlapped but at least 5 symbols are left for the digital phase lock loop (PLL) to lock the reception.
A LoRa preamble consists of a fixed part (i.e., synchronization word of 2 symbols, and an additional 2.25 symbols) and a configurable part that varies from 6 to 65535 symbols. In summary, we make the following useful observations:
\begin{itemize}
    \item Inter-SF interference can be neglected if the SIR is not too low.
    \item Unless the smallest SFs are used by an extremely large number of EDs, in studying the co-SF interferers only the dominant can be considered. 
    \item LoRa can survive collisions affecting the first few symbols in the preamble even if the SIR is lower than the minimum required for the power capture to take place.
\end{itemize}
Based on these observations, we assume perfect orthogonality among SFs and consider the dominant co-SF interference only in the subsequent analytical modeling. Arguably, the assumption becomes weak for a very large number of interferers. However, as the \mbox{S-ALOHA} and CSMA reduce the number of potential interferers (see Sec.~\ref{sec:IPPP}), it is more plausible to consider dominant than cumulative interference.

\subsection{MAC-dependent Intensity of Point Process of Interferers}
\label{sec:IPPP}
To characterize interference, be it cumulative or dominant, a measure of interest is the set of all interfering devices to a reference transmission. This interference set, a subset of the original PPP, is 
shaped according to the traffic intensity and the medium access mechanism 
used by the network. Motivated by the discussion in the previous section, our interest is to characterize the dominant interference $\mathcal{I}^*$ at the 
gateway, i.e., 
\begin{equation}
\label{eq:Iq}
\displaystyle
	\mathcal{I}^*=\max\limits_{i\in  \phi_\text{I}}\left(I_i\right) = p_\text{tx
}\max\limits_{i\in  \phi_\text{I}}\left\{g_i l\left(r_i\right)\right\}\text{,}
\end{equation}
where $\phi_I$ is the PP of interferers (interference geometry), while $g_i$ 
is the fading coefficient and $l\left(r_i\right)$ the path loss function of 
the $i$-th interferer at distance $r$ from the gateway.

To evaluate the outage or success probability, the distribution of interference must be known \cite{haenggi2009interference}. The interference distribution for 
$\mathcal{I}^*$, $F_{\mathcal{I}^*}(x)$, can be determined from extreme order statistics. That is, if 
$F_I(x)$ is the distribution of $I_i=g_i l(r_i)$, then $F_{\mathcal{I}^*}(x)$ for interferers of random size is
\begin{equation}
\label{eq:F_I}
F_{\mathcal{I}^*}(x) = \sum_n \Pr\left(N=n\right)\left[F_{X_{i}}(x)\right]^n\text{,}
\end{equation}
where $\Pr\left(N=n\right)$ is the PMF of the arrival process of 
interfering transmissions. For Poisson arrivals with parameter $\mu$, 
$\Pr\left(N=n\right) = \mu^k e^{-u}/k!$, while $\mu$ for an annulus of radii $r_1$ and $r_2$ is defined as
\begin{equation}
	\mu=\mathcal{V} \pi  \left(r_2^2-r_1^2\right)\text{.}
	\label{eq:mu}
\end{equation}

In \eqref{eq:mu}, $\mathcal{V}$ is the intensity of PP of interferers $\phi_I$, 
defined by the channel access rules. Therefore, the intensity $\mathcal{V}$ is the 
parameter that changes the interference distribution in \eqref{eq:F_I}, and 
eventually, the system performance with a change in MAC. In what follows, we find the intensity $\mathcal{V}_x$ of the PP of interferers $\phi_I$ for $x\in\{\textrm{P-ALOHA}, \textrm{S-ALOHA}, \textrm{CSMA}\}$. The PP $\phi_I$ and its intensity $\mathcal{V}$, in essence, is obtained by applying channel access-dependent thinning of the EDs' PPP $\phi_\text{ED}$.

\subsubsection{P-ALOHA}
\label{sec:aloha}
P-ALOHA is a natural match for the PPP since it maintains the distributional properties of the PPP. As a result, if the EDs
form a PPP of intensity $\lambda$ and transmit independently with probability
$\alpha$, the active set of transmitters, by independent thinning property of PPP~\cite{haenggi2009interference}, forms a PPP of intensity $\alpha \lambda$.
In \mbox{P-ALOHA}, an ED with a message can immediately transmit and the vulnerability time of the transmission is two times the message time-on-air (ToA). Therefore, the intensity of the interferers PP is 
\begin{equation}
\displaystyle
\label{eq:v_aloha}
	\mathcal{V}_{\textrm{P-ALOHA}}=2 \alpha\lambda.
\end{equation}

\subsubsection{LoRA}
\label{sec:LoRa_aloha}
LoRa access method is \mbox{P-ALOHA} but as analyzed in Sec.~\ref{sec:int_model}, LoRa is immune to interference affecting the first few symbols in a message preamble. This effect leads to reduction in the vulnerability time of a message by $(T_p-5T_\text{sym})/(\textrm{ToA})$, where $T_p$ is the preamble duration and $T_\text{sym}=2^\text{SF}/B$ is the symbol time. While this effect is negligible for large messages, its contribution becomes visible at small payloads (i.e., a few Bytes). The intensity of the PP of the interferers for LoRa is 
\begin{equation}
\displaystyle
\label{eq:v_lora}
	\mathcal{V}_{\textrm{LoRa}}=\left[\left(2-\frac{T_p-5T_\text{sym}}{\text{ToA}}\right) \cdot \alpha\cdot \lambda\right]<\mathcal{V}_{\textrm{P-ALOHA}}.
\end{equation}
\subsubsection{\mbox{S-ALOHA}}
In \mbox{S-ALOHA}, devices can transmit only at the border between two timeslots.  The vulnerability time of a transmission is equal to ToA. To be effective, \mbox{S-ALOHA} requires synchronization of all the devices with a common source. A synchronization error between devices can cause inter-slot collisions, that can be reduced by introducing a guard interval (GI).  The probability that transmissions in adjacent slots collide depends on the timing error distribution and the size of the GI. The intensity of the interferers $\mathcal{V}_{\textrm{S-ALOHA}}$ can be defined as 
\begin{equation}
\displaystyle
\label{eq:v_saloha}
	\mathcal{V}_{\textrm{S-ALOHA}}= \left(1+g\right)\left(1+p_L+p_R\right)\cdot\alpha\cdot \lambda,
\end{equation}
where $g$ is the ratio between the length of the GI $T_g$ and the ToA, $p_L$ and $p_R$ are the probability of experiencing inter-slot interference respectively from  the previous and successive.
For a Gaussian distributed timing error with standard deviation $d$, the probability of inter-slot collision with a message transmitted in the previous and next timeslot  are  derived from \cite{crozier1990sloppy} respectively  as
\begin{align}
p_L&=Q\left(\frac{T_g+T_p-5T_\text{sym}}{\sqrt{2}d}\right),   &   p_R&=Q\left(\frac{T_g}{\sqrt{2}d}\right).     
\end{align}

\label{sec:saloha}
Collisions caused by messages transmitted in the previous (i.e., left) timeslot are less likely since they affect the preamble of the transmitted message as highlighted in Sec.~\ref{sec:int_model}.
\subsubsection{NP-CSMA}
\label{sec:csma}
The use of CSMA can reduce the number of interfering devices by employing the LBT technique, where each ED senses the channel to determine if it is free before attempting a transmission. We consider an asynchronous \mbox{non-persistent} CSMA so that the EDs can save energy by avoiding synchronization to the gateway. The operation of an ED using the \mbox{NP-CSMA} protocol to transmit a message to the gateway can be described as follow: 1) the ED senses the channel and transmits immediately if it is free; 2) if the ED senses a busy channel, the transmission is rescheduled to a later time according to a random delay distribution. At this new point, the ED senses the channel again, repeating the process until it transmits the message.  
 During the backoff, the ED can go in an energy-saving mode as it is not required to sense the channel. 
Under these access rules, the interferers are not distributed according to a PPP but instead can be modeled as a  Mat\`ern hard-core process of Type II (MHCPP-II)~\cite{baccelli2010stochastic}. To use this model, it is necessary to assume that the probability of two EDs in the same contention domain having the same backoff counter value is negligible. This requires mechanisms that  alleviate the backoff timer ties problem to achieve an intact CSMA spectrum access, for example, by using adaptive contention window size techniques.
In an  MHCPP-II, each device is marked with a mark uniformly distributed  in $\left[0,1\right]$, representing the random backoff value selected by the device.
The MHCPP-II retains in a neighborhood only the device with the smallest mark; that is, in a contention domain, only the device with the lowest backoff counter can transmit.

\begin{prop}
The intensity of the interfering EDs for asynchronous NP-CSMA is given by 
\begin{equation}
\mathcal{V}_{\textrm{CSMA}}=\left(2-\frac{T_p-5T_\text{sym}}{\textrm{ToA}}\right)\left(1-H\right)\tilde{\lambda},
\label{eq:v_csma}
\end{equation}
where, $H$ is the probability that an ED is located within transmitter contention domain and $\tilde{\lambda}$ is the intensity of the transmission of the interfering EDs. 
\end{prop}
\begin{proof}
In asynchronous CSMA, the transmissions and backoffs are not aligned to timeslots. Therefore, the vulnerability time of a transmission is twice the ToA, and the first term in \eqref{eq:v_csma} is the same as in \eqref{eq:v_lora}. The second term is the probability that the generic ED in the annulus is located outside the contention domain of the transmitter. 

The fraction of EDs that are located inside the transmitter contention domain can be found from \cite{baccelli2010stochastic} as
\begin{equation}
\begin{aligned}
H&= \int_{0}^{2r_2} e^{-\frac{P_0}{p_\text{tx}l\left(r\right)} } f_{R}\left(r\right)dr ,
\end{aligned}
\label{eq:H}
\end{equation}
where $e^{-\frac{P_0}{P_\text{TX}l\left(r\right)} }$ is the probability that in a unit mean Rayleigh fading and path loss function $l\left(\cdot\right)$, the received signal strength at a distance $r$ from a transmitter is above the detection threshold $P_0$. The distance distribution between two random device using the same SF is $f_{R}(r)$. If the ED are uniformly distributed on a disk of radius $R$, $f_{R}(r)$ is given by \cite{mathai1999introduction}
\begin{equation}
f_{R}(r)=\frac{4r}{\pi R^2}\left[\cos^{-1}\left(\frac{r}{2R}\right)-\frac{r}{2R}\sqrt{1-\frac{r^2}{4R^2}}\right].
\end{equation}
In other geometries, such as in the case of annuli, the expression of the distance distribution becomes more complex.
\end{proof}
The expected number of devices in the contention domain of the transmitter can be obtained from \eqref{eq:H} as $E\left[n\right]=\lambda\pi \cdot \left(r_2^2-r_1^2\right) \cdot H$.
The expression of $\tilde{\lambda}$ must be derived from the analysis of the distribution of the EDs outside of the contention domain of the reference transmitter. The intensity of the MHCPP-II $\tilde{\lambda}$  can be found by the thinning process of the original PPP intensity $\lambda$ by the probability that an ED other than the reference transmitter has a message for the gateway and successfully contended for the channel. Due to the complexity of handling MHCPP-II, we approximate the MHCPP-II of the active interferers as an equi-dense PPP; that is, the interferers are uniformly and randomly distributed in the model.
\begin{lemm}
The intensity, $\tilde{\lambda}$, of the interfering  EDs is approximated with an MHCPP-II of intensity 
  
\begin{equation}
\label{eq:lambda_hat}
\tilde{\lambda}=\lambda p\frac{1-e^{-E\left[n_{A}^*\right]}}{E\left[n_{A}^*\right]}\text{,}
\end{equation}
where  $E\left[n_{A}^*\right]$ is the expected number of neighbors for one interfering ED and $p$ is the probability that an ED is contending for channel access.
\end{lemm}
\begin{proof}
Not all the EDs in a contention domain will be actively trying to transmit, hence we introduce a probability $p$ that a device  is contending for channel access. The intensity of the EDs that have a message for the gateway is $\lambda p$.

Let  $E\left[n_{A}^*\right]$ be the expected number of neighbors for one interfering ED. An ED  is able to transmit only if it has the smallest marker among all the active EDs in its contention domain. The probability for this event to occur can be found according to \cite{baccelli2010stochastic} as $\frac{1-e^{-E\left[n_{A}^*\right]}}{E\left[n_{A}^*\right]}$.
Given that all active EDs within the contention domain of the reference transmitter are blocked from transmitting, the expected number of neighbors for an interfering node is upper bounded by 
\begin{equation}
\label{eq:bound}
E\left[n_{A}^*\right] \leq \begin{cases}\lambda \pi \left(r_2^2-r_1^2\right)\cdot \left[1-H\right]p &\text{ if }  H\geq 0.5\text{;}\\
\lambda \pi \left(r_2^2-r_1^2\right)\cdot Hp &\text{ if }  H<0.5\text{.}\end{cases}
\end{equation}
Depending on the detection threshold $P_0$ and the annulus geometry, not all EDs outside of the transmitter contention domain  will be part of the same interference contention domain; that is, for a transmitter, multiple interfering ED can be present at the same time.
To find the expected number of active EDs within the contention domain of an interfering ED, we scale the upper bound of the number of neighbor by a smooth function $g(H)$.  The expected number of neighbors for an interfering ED is 
\begin{equation}
\label{eq:En*}
E\left[n_{A}^*\right]\simeq p\pi \lambda \left(r_2^2-r_1^2\right) \cdot \min \big( 1- H,H \big)\cdot \left(g(H)\right)^{-1}\text{.}
\end{equation}
The derivation of $g\left(H\right)$ is presented in Appendix \ref{sec:Appendix}.
\end{proof}
Whereas in \mbox{P-ALOHA} and \mbox{S-ALOHA} the transmission probability can be assumed to be equal to the activity factor $\alpha$, in CSMA it has to be determined.  An ED has to sense the channel to be free before transmitting, so that under the same traffic load, the probability that an ED is contending for channel access is larger than the transmission probability of \mbox{P-ALOHA} or \mbox{S-ALOHA} ($p\geq \alpha$).
To find the probability $p$ that an ED is contending for channel access, we model each ED as a \textit{Geo/Geo/1} queue without buffer.
It implies that an ED attempts to transmit only the most recent of the generated messages. Both the generation and transmission of messages at an ED follow a geometric distribution with probability $p_A=\alpha$ and $p_D$, respectively. The transition probability matrix $\mathbf{P}$ of the \textit{Geo/Geo/1} queue is given by
\begin{equation}
\mathbf{P}=  \bigg[\begin{smallmatrix}
   1-p_A & p_A\\[1mm]
   \left(1-p_A\right)p_D& 1-p_D + p_A p_D
   \end{smallmatrix}\bigg]
\label{eq:Ptran}
\end{equation}
The probability that an ED has a message is the probability that the queue is non-empty. By finding the steady state vector $\left[\pi_0, \pi_1\right]$ of the probability matrix in (\ref{eq:Ptran}), $p$ is determined as
\begin{equation}
\label{eq:p}
p=\pi_1=\left(1+\frac{\left(1-p_A\right)p_D}{p_A}\right)^{-1}.
\end{equation}
Without acknowledgments, a departure from the queue is equivalent to a transmission. The departure probability is
\begin{equation}
p_D=H\frac{1-e^{-E\left[n_A\right]}}{E\left[n_A\right]}+\left(1-H\right)\frac{1-e^{-E\left[n_A^*\right]}}{E\left[n_A^*\right]}\text{,}
\end{equation}
where $E\left[n_{A}\right]=E\left[n\right]\cdot p$ is the expected number of active neighbours for the transmitting ED.

The intensity of the MHCPP-II of interferers $\mathcal{V}_{\textrm{CSMA}}$ is found  from \eqref{eq:v_csma} in Proposition 1 by using \eqref{eq:H}, \eqref{eq:lambda_hat} and \eqref{eq:p}.

\subsection{Derivation of CCDFs of SNR and SIR}
\label{sec:SNR&SIR}
We model the wireless channel as a Rayleigh block-fading with additive white Gaussian noise (AWGN) and free space path loss. The
variance of the AWGN is $\sigma^2 [\textrm{dBm}]=-174+\textrm{NF}+10\log_{10}\textrm{BW}$, where NF is the noise figure of the receiver and BW is the channel bandwidth.
Messages are transmitted with power $p_\text{TX}$ and suffer from path-loss modeled by the power-law function $l\left(r\right)=\gamma r^{-\beta}$, where $\gamma$ and $\beta$ are respectively the frequency dependent factor and the path-loss exponent. For the carrier wavelength $\lambda_c$, $\gamma=\left(\lambda_c/4\pi\right)^2$.

If a message is transmitted by a device located at a distance $r$ from the gateway, the received power is given by $
 p_\text{rx}\left(r\right)=p_\text{tx} h l(r)$,
where $h$ is the channel gain between the transmitter and the gateway. The probability that the SNR of the received message at the gateway is above the threshold $\theta$ is
\begin{equation}
\label{eq:Psnr}
\begin{aligned}
\displaystyle
	\mathbb{P}\left[\text{SNR}(r)\!\geq \theta\right]\!=\!\mathbb{P}\left[h\geq \frac{\sigma^2 \theta}{p_\text{tx} l(r)}\right] =\text{exp}\left(\!- \frac{\sigma^2\theta}{p_\text{tx} l(r)} \right)\text{,}
\end{aligned}
\end{equation}
obtained by the fact that $h\sim \text{exp}(1)$.

As discussed in Sec.~\ref{sec:int_model}, the receiver might still be able to correctly decode the message depending on the relative signal strength between the signal and the interferers. 
At the gateway, the SIR for a signal transmitted by an ED at distance $r$ is $\text{SIR}(r)={p_\text{rx}(r)}/{\mathcal{I}^*}$,

where $\mathcal{I}^*$ is the dominant interference given by (\ref{eq:Iq}). The probability that the SIR of a message under co-SF interference is above the threshold $\delta$ is given by 
\begin{equation}
\begin{aligned}
\displaystyle
	&\mathbb{P}\left[ \text{SIR}(r)\geq \delta\right]=\mathbb{P}\left[ \mathcal{I}^* \leq \frac{p_\text{tx} h l\left(r\right)}{ \delta}\right]\text{.}
	\label{eq:CCDF_SIR}
\end{aligned}
\end{equation}

Using the series $\exp(x)=\sum_{k=0}^{\infty}x^k/k!$ in \eqref{eq:F_I}, and taking the expectation over the channel gain $h$, \eqref{eq:CCDF_SIR} becomes 
\begin{equation}
\label{eq:Psir2}
\begin{aligned}
\displaystyle
	&\mathbb{P}\left[ \text{SIR}(r)\geq \delta\right]\!=\! e^{-\mu}\!\!\int_0^\infty\!\!\!\!\! e^{-z}\!\exp\left(\!\mu F_{X_{i}}\left(\frac{z l\left(r\right)}{\delta}\right)\!\!\right) \!dz\text{,}
\end{aligned}
\end{equation}
where $F_{X_{i}}\left(\cdot\right)$ is the CDF of the product of the probability distributions of $l(r)$ and $h$, obtained from \cite[(10)]{mahmood2019scalability}.
\subsection{Performance Metrics}
\label{sec:perf_metrics}
In this section, we derive the performance metrics used to compare the random access protocols.

\subsubsection{Success and Coverage Probability}
A transmission from an ED at distant $r$ from the gateway is successful only if both the SNR and SIR conditions are satisfied at the gateway. The success probability for an ED transmitting using SF~$q$ is
\begin{equation}
\displaystyle
\label{eq:ps}
\begin{aligned}
	p_{\text{succ},q}\left(r\right)&=\mathbb{P}\big[\left\lbrace \text{SIR}(r)\geq \delta\right\rbrace\cap\left\lbrace \text{SNR}(r)\geq \theta_q\right\rbrace\big]\\
	&\geq \mathbb{P}\left[ \text{SIR}(r)\geq \delta\right]\cdot \mathbb{P}\left[\text{SNR}(r)\geq \theta_q\right]\text{,}
\end{aligned}
\end{equation}
where $\mathbb{P}\left[\text{SNR}(r)\geq \theta_q\right]$ and $\mathbb{P}\left[ \text{SIR}(r)\geq \delta\right]$ are given by (\ref{eq:Psnr}) and (\ref{eq:Psir2}), respectively. 
The inequality comes from the fact that the probability that both the SNR and SIR conditions are satisfied are not independent. For instance, a message that arriving at the gateway satisfies the SIR condition, it is more likely to also satisfy the SNR condition. 

Let $l_q$ and $l_{q-1}$, respectively, be the outer ($r_2$) and inner radius ($r_1$) of the annulus containing all EDs using SF~$q$. The coverage probability of an ED in the service area is given by
\begin{equation}
\displaystyle
\label{eq:Pcov}
	p_\text{cov}=\sum_{q}\int_{l_{q-1}}^{l_{q}} p_{\text{succ},q}\left(r\right)\cdot f_\text{ED}(r) dr.
\end{equation}

\subsubsection{Channel Throughput}
We calculate the channel throughput as the product of the offered traffic and the success probability. Because of the assumption in Sec.~\ref{sec:int_model}, we consider orthogonal SFs; hence the channel throughput for SF~$q$ is given by
\begin{equation}
\displaystyle
\label{eq:Th}
	S_q=\alpha\lambda \int_{l_{q-1}}^{l_{q}} p_{\text{succ},q}(r) \cdot f_\text{ED}(r) dr.
\end{equation}

\subsubsection{Energy Efficiency}
\label{subsect:ee}

Let $P_\text{TX}$ and $P_\text{RX}$,  respectively, be the power consumption of an ED while transmitting and receiving/sensing the channel. Assuming that the energy consumption in sleep mode is negligible, the energy efficiency of an ED at a distance $r$ from the gateway can be found as
\begin{equation}
\displaystyle
\label{eq:eta}
	\eta\left(r\right)=\frac{p_{\text{succ},q}\left( r\right)\cdot \text{Payload}}{E_0}\text{,}
\end{equation}
where Payload is the payload size of a message, and $E_0$ is the energy used to deliver a message to the gateway, which depends on employed channel access method in the network.

In \mbox{P-ALOHA}, since all the energy is used by a device for transmitting its messages, $E_0$ is
\begin{equation}
\displaystyle
\label{eq:E_0}
	E_0=P_\text{TX}\text{ToA}\text{.}
\end{equation}

In \mbox{S-ALOHA}, additional energy is required for each transmission to maintain the synchronization with the gateway. We assume that the synchronization is maintained by periodic beacons of duration $T_B$ transmitted by the gateway every $T_\text{SYNC}$. In comparison with the synchronization method using  two-way message exchange of \cite{polonelli2019slotted}, the synchronization by periodic beacons assumed  in this analysis reduces the communication overhead in the network.
\begin{equation}
\displaystyle
\label{eq:nuS}
	E_0=P_\text{TX}\text{ToA}+P_\text{RX}\cdot T_B \cdot \left(\frac{\text{ToA}}{\alpha T_\text{SYNC}}\right),
\end{equation}
where the additional term is the energy used for synchronizing the EDs.

The energy efficiency for NP-CSMA can be determined as
\begin{equation}
\displaystyle
\label{eq:nuc}
	E_0=P_\text{TX}\text{ToA}+P_\text{RX}\frac{E\left[n_A\right]}{1-e^{-E\left[n_A\right]}} T_\text{CAD},
\end{equation}
where the second term is the average energy used by CAD, obtained by multiplying the expected number of CAD attempts per transmission with the CAD duration $T_\text{CAD}$. Since messages are not acknowledged by the GW,  retransmissions are not considered when calculating the energy efficiency. Moreover, since the EDs enter energy-saving mode during their backoff, we assume that the energy consumption during the random backoff is negligible.

\section{Results and Discussion}
\label{sec:Results_Analysis}
The analytical model was validated using Monte-Carlo simulations. The success probability at each distance is an average value of 2000 simulation runs, while the error bars in the figures correspond to a confidence interval of 95\%.
The studied channel access techniques \mbox{P-ALOHA}, \mbox{S-ALOHA} and \mbox{NP-CSMA}, were compared in terms of success probability, energy efficiency, and channel throughput. The analysis is based on typical parameters of LoRa for an outdoor monitoring application, given in Table~\ref{tb:par}. Unless explicitly stated, the EDs can belong to any of the three classes defined in LoRaWAN, class A, B and C.
\begin{table}[!htb]
\caption{LoRa Parameters}
\centering
\scalebox{0.9}{\setlength\tabcolsep{3 pt}
\begin{tabular}{lllllll}
\cline{1-3} \cline{5-7}
\multicolumn{1}{|l|}{\textbf{Parameter}} & \multicolumn{1}{l|}{\textbf{Sym.}} & \multicolumn{1}{l|}{\textbf{Value}} & \multicolumn{1}{l|}{} & \multicolumn{1}{l|}{\textbf{Parameter}} & \multicolumn{1}{l|}{\textbf{Sym.}} & \multicolumn{1}{l|}{\textbf{Value}} \\ \cline{1-3}\cline{1-3} \cline{5-7} \cline{5-7}  
\multicolumn{1}{|l|}{Bandwidth} & \multicolumn{1}{l|}{$B$} & \multicolumn{1}{l|}{$125$ kHz} & \multicolumn{1}{l|}{} & \multicolumn{1}{l|}{Noise PSD} & \multicolumn{1}{l|}{$N_0$} & \multicolumn{1}{l|}{$-174$ dBm/Hz} \\
\multicolumn{1}{|l|}{Carrier Frequency} & \multicolumn{1}{l|}{$f_c$} & \multicolumn{1}{l|}{$868$ MHz} & \multicolumn{1}{l|}{} & \multicolumn{1}{l|}{Noise Figure} & \multicolumn{1}{l|}{NF} & \multicolumn{1}{l|}{$6$ dBm } \\
\multicolumn{1}{|l|}{Transmit Power} & \multicolumn{1}{l|}{$p_\text{tx}$} & \multicolumn{1}{l|}{$14$ dBm} & \multicolumn{1}{l|}{} & \multicolumn{1}{l|}{Activity Factor} & \multicolumn{1}{l|}{$\alpha$} & \multicolumn{1}{l|}{$0.33$ \%} \\
\multicolumn{1}{|l|}{Pathloss Exponent} & \multicolumn{1}{l|}{$\beta$} & \multicolumn{1}{l|}{3} & \multicolumn{1}{l|}{} & \multicolumn{1}{l|}{Coding Rate} & \multicolumn{1}{l|}{} & \multicolumn{1}{l|}{4/8} \\
\multicolumn{1}{|l|}{Message Preamble} & \multicolumn{1}{l|}{} & \multicolumn{1}{l|}{8 Symbols} & \multicolumn{1}{l|}{} & \multicolumn{1}{l|}{Payload} & \multicolumn{1}{l|}{} & \multicolumn{1}{l|}{10 Bytes} \\
\multicolumn{1}{|l|}{Clock Skew} & \multicolumn{1}{l|}{} & \multicolumn{1}{l|}{40ppm} & \multicolumn{1}{l|}{} & \multicolumn{1}{l|}{Guard Interval} & \multicolumn{1}{l|}{$T_g$} & \multicolumn{1}{l|}{10.24 ms} \\\cline{1-3} \cline{5-7} 
 &  &  &  &  &  & 
\end{tabular}}
\vspace{-10pt}
\label{tb:par}
\end{table}
The SFs are assigned to the uniformly distributed EDs in ascending order according to the distance from the gateway. Unless otherwise stated, the SFs assignment follows an equal-interval-based allocation, generating annuli of equal width. The SNR thresholds are $\theta_q\in\left\lbrace-6,-9,-12,-15,-17.5,-20\right\rbrace$ dB for $q=7,\ldots,12$, and the SIR threshold $\delta=1$ dB \cite{georgiou2017low}. 
\subsection{Success and Coverage Probability}
\subsubsection{LoRa and S-LoRa}

In this section, the performance of perfectly synchronized (ideal) \mbox{S-ALOHA} is compared to an implementation of slotted ALOHA in LoRa (\mbox{S-LoRa}).
A discussion on how to provide synchronization in LoRa is not in the scope of this work.  In studying S-LoRa, we assume that all EDs are synchronized using periodic beacons from the gateway and belong to either class~B or class~C. The synchronization error is assumed to be normally distributed with a standard deviation equal to the product of clock skew and the synchronization interval $T_\text{SYNCH}$.

\begin{figure}[!htp]
\centering
\includegraphics[width=0.9\columnwidth]{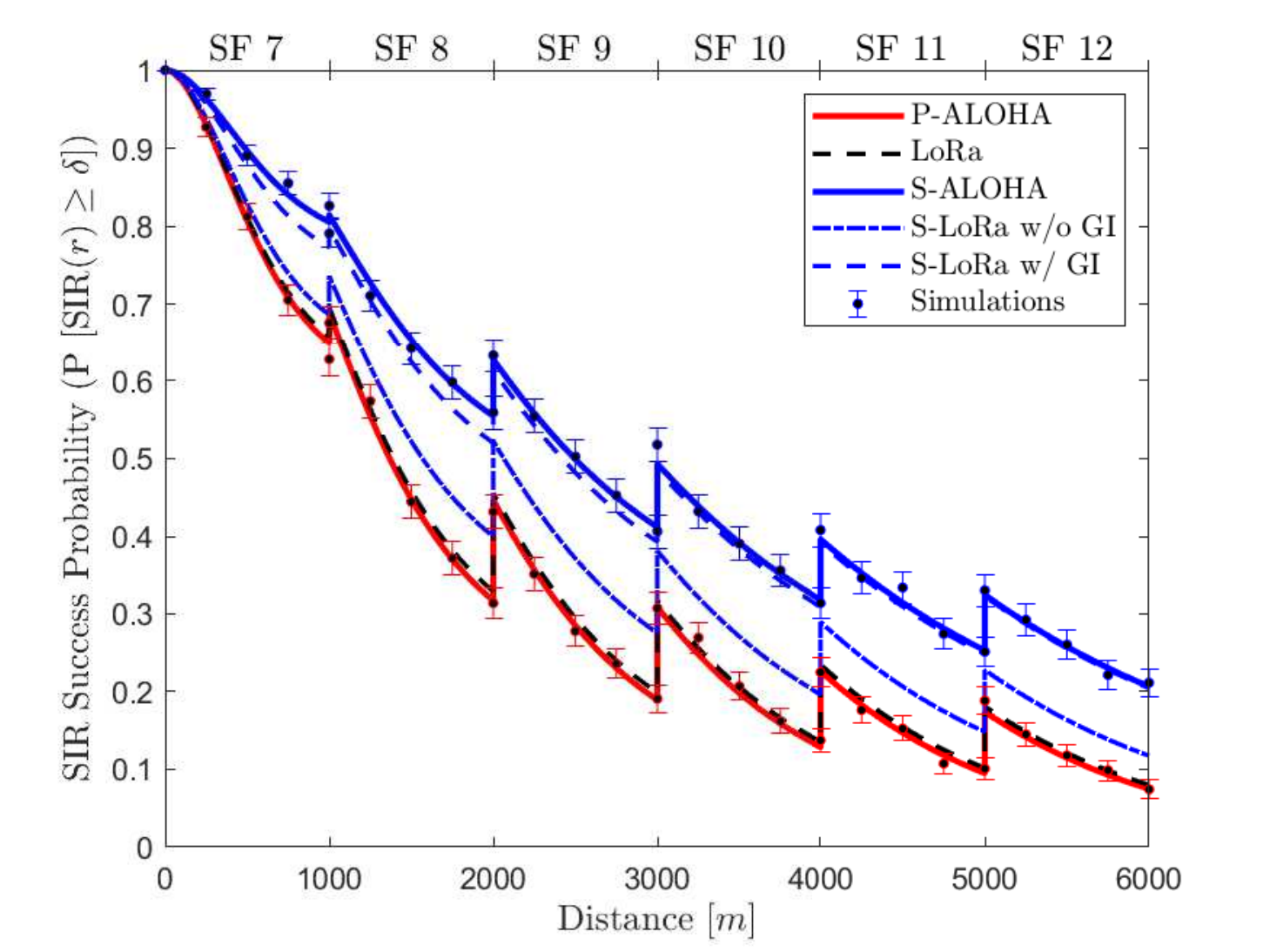}
\caption{SIR success probability of LoRa, s-LoRa, \mbox{P-ALOHA} and \mbox{S-ALOHA} $\bar{N}=3000$.}

\label{fig:p_succ_1}
\end{figure}

\begin{figure}[!htp]
\centering
       \includegraphics[width=0.90\columnwidth]{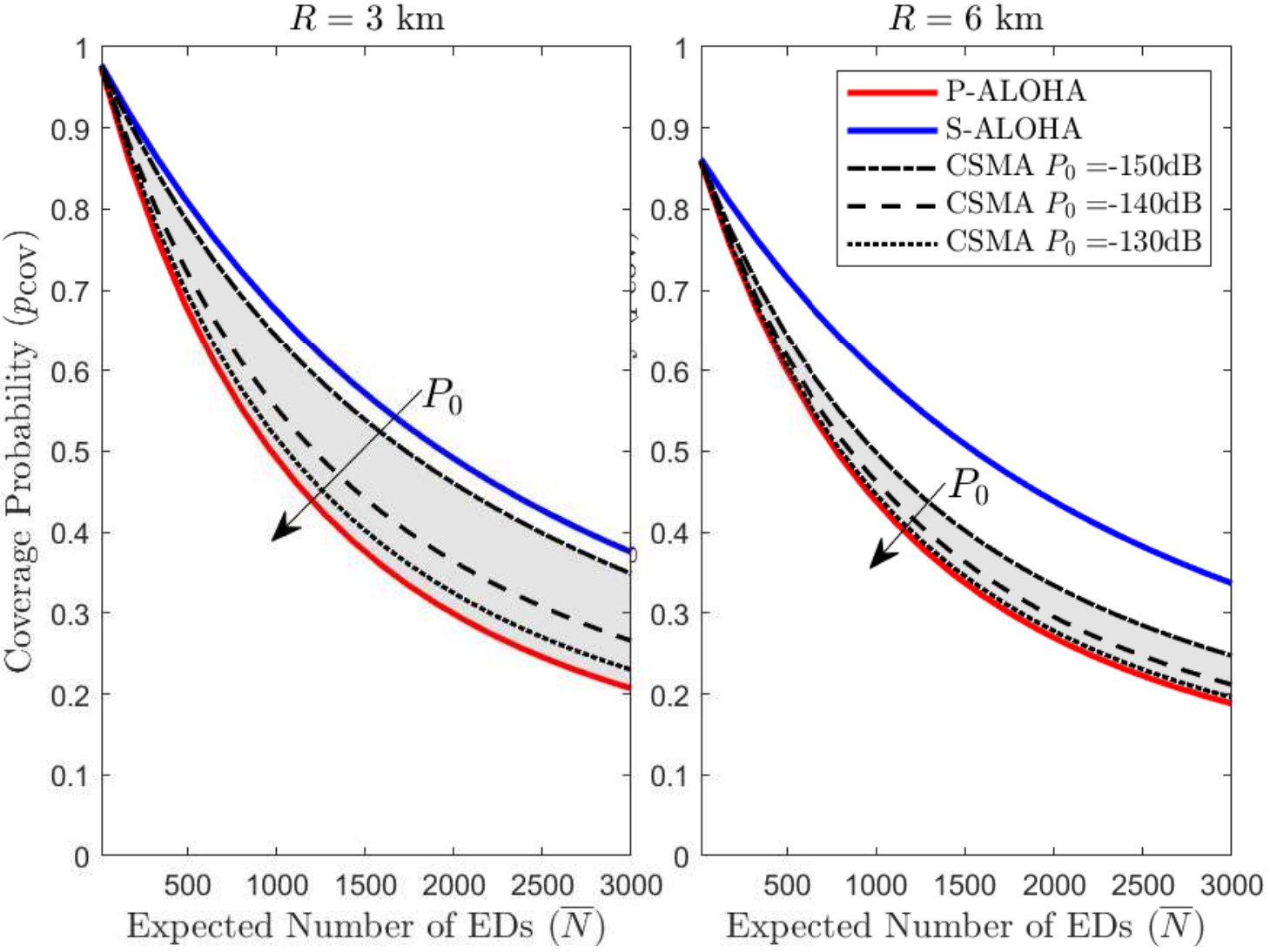}
     \caption{Coverage probability of \mbox{P-ALOHA}, \mbox{S-ALOHA} and \mbox{NP-CSMA}, \mbox{$\bar{N}=3000$}.}
     \vspace{-15pt}
     \label{fig:p_cov}
\end{figure}

\begin{figure*}[!htp]
    \centering
     \subfloat[Equal-interval-based allocation\label{fig:p_succ_2a}]{%
       \includegraphics[width=0.44\textwidth]{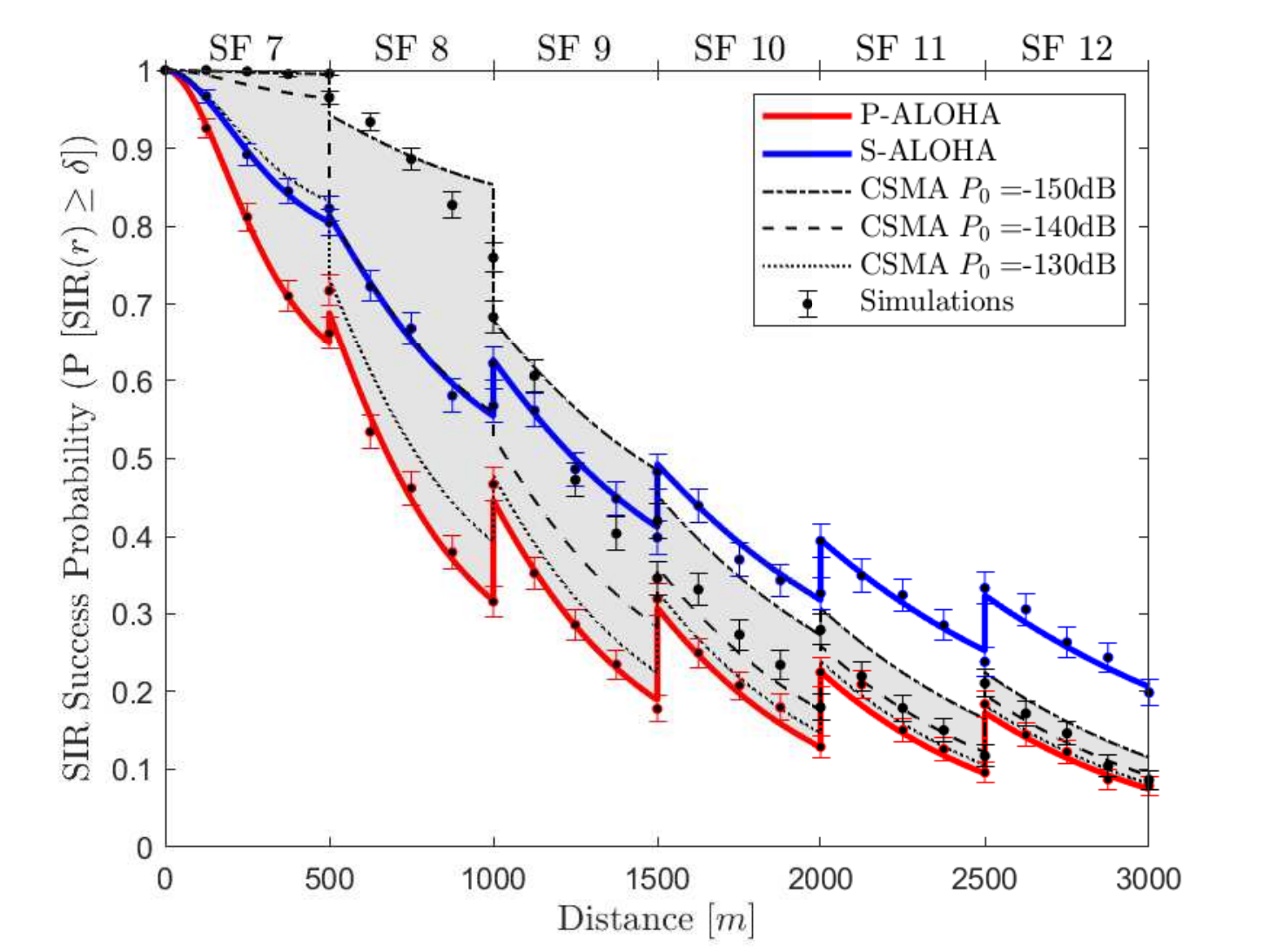}
     }
     \hspace{45pt}
     \subfloat[Equal-area-based allocation\label{fig:p_succ_2b}]{%
       \includegraphics[width=0.44\textwidth]{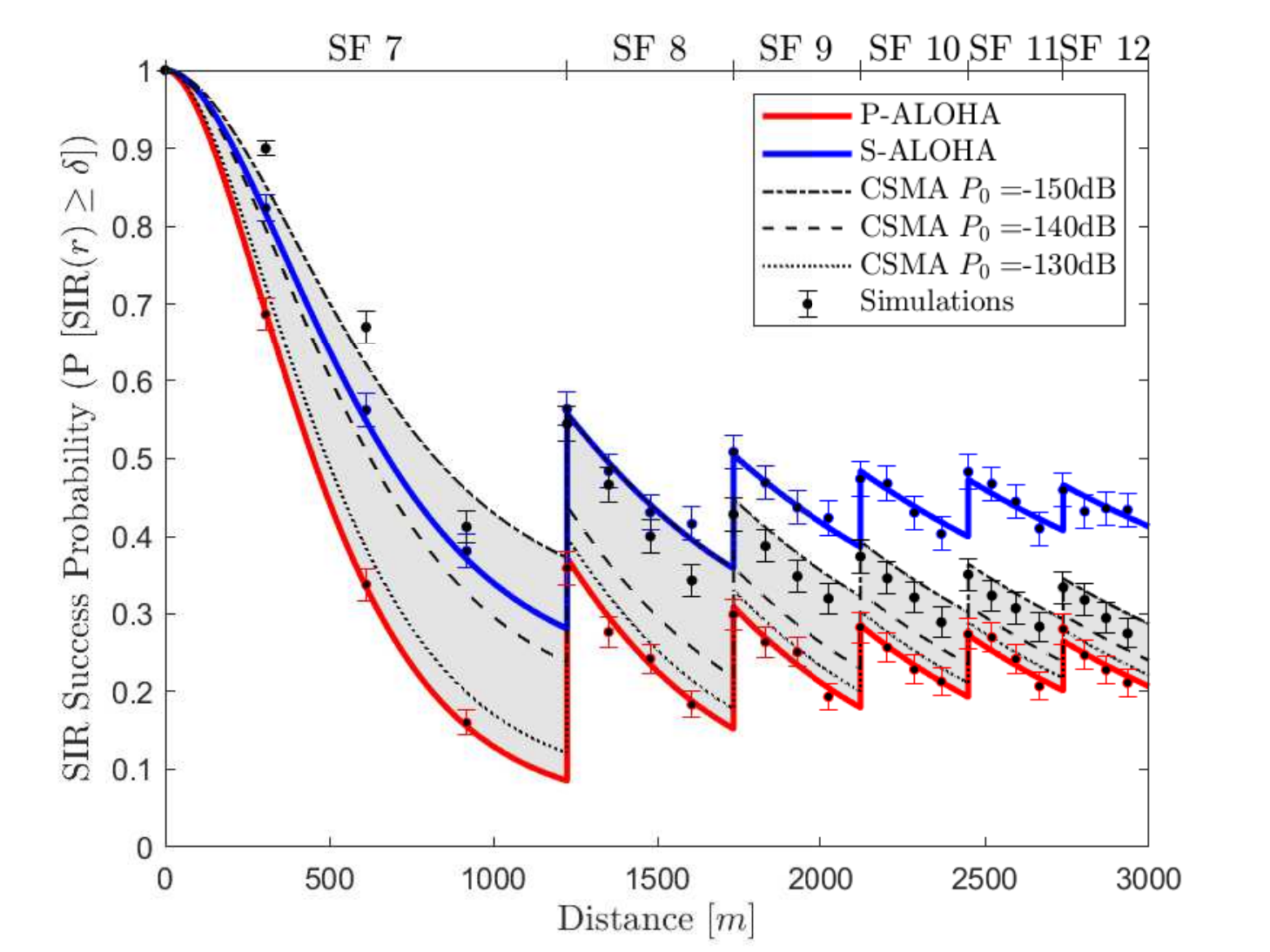}
     }
     \caption{SIR success probability of \mbox{P-ALOHA}, \mbox{S-ALOHA} and \mbox{NP-CSMA}, $\bar{N}=3000$.}
     \label{fig:p_succ_2}
     \vspace{-12pt}
\end{figure*}
Fig.~\ref{fig:p_succ_1} shows SIR success probability obtain from \eqref{eq:Psir2} when $\bar{N}=3000$ for \mbox{P-ALOHA}, LoRa, \mbox{S-ALOHA}, \mbox{S-LoRa} with GI (S-LoRa w/ GI)  and \mbox{S-LoRa} without GI (S-LoRa w/o GI). The figure shows the characteristic saw-tooth  behavior of the success probability in LoRa, as presented and analyzed before in \cite{georgiou2017low} and \cite{mahmood2019scalability}.
The shorter vulnerability time of \mbox{S-ALOHA} gives a clear advantage over \mbox{P-ALOHA}, reducing the interferers by half and increasing the success probability for all SFs on average by 0.16. 
Accounting for the capture effects of LoRa shows a small performance improvement over \mbox{P-ALOHA}. In \mbox{S-LoRa}, when synchronization error is considered, it shows a decrease in the SIR success probability compared to the ideal case (\mbox{S-ALOHA}). When the GI is removed (i.e., \mbox{S-LoRa} w/o GI), the success probability of \mbox{S-LoRa} does not degrade to the performance of LoRa, since a part of the LoRa preamble acts as a GI versus collisions with messages transmitted in the previous timeslot. 

Fig.~\ref{fig:p_cov} shows the coverage probability of LoRa and \mbox{S-LoRa} when the EDs are located within a typical distance of $3$ and $6$ km from the gateway.  The coverage probability decreases as the number of EDs, and hence the interference, increases. LoRa and \mbox{S-LoRa} can offer similar coverage probabilities. However, \mbox{S-LoRa} offers better scalability for an increasing number of EDs. 

\subsubsection{NP-CSMA}
Fig.~\ref{fig:p_succ_2} compares the SIR success probability of CSMA-based LoRa with that of P-ALOHA and S-ALOHA for equal-interval-based and equal-area-based SF allocation. 
In CSMA, the SIR success probability depends on the detection threshold used by the sensing mechanism. In this analysis, we considered sensing thresholds  $P_0=\left\lbrace-150, -140, -130\right\rbrace$ dB representing extreme operating range of the CAD mechanism of LoRa chipset \cite{CAD}. 
High sensing thresholds cause the success probability of CSMA to degrade to the performance of \mbox{P-ALOHA}, whereas low sensing threshold increases the success probability of CSMA above that of \mbox{S-ALOHA}. The CSMA performance is affected  by the size of the area in which the same SF EDs are located, making the analysis of different SF allocation methods particularly interesting. Devices in the outer annuli are affected more by the hidden terminal problem because of their longer average distance. However, using a low sensing threshold to combat the hidden terminal problem may lead to additional backoff by the EDs located in the outer annuli due to interference from neighboring gateways. 

From Fig.~\ref{fig:p_succ_2}, it can be argued that for the outer annuli, even in our conservative model with a single gateway, \mbox{NP-CSMA} offers little gain over \mbox{P-ALOHA}. Hence CSMA should only be used  by the EDs with the smallest SFs, located closest to the gateway and to each other.
In the case of equal-interval-based allocation, if the distance between EDs is smaller than the sensing range, no hidden terminals are present, and all the EDs have the same success probability regardless of their position within the annulus (SF~7 in  Fig.~\ref{fig:p_succ_2a}). 
Conversely, in the case of equal-area-based allocation (Fig.~\ref{fig:p_succ_2b}), the difference between the outer and inner radii of the smaller SFs increases significantly, resulting into  a more steep decrease of the success probabilities. The benefit of using channel sensing for small SFs, clearly visible for equal-interval-based allocation in Fig.~\ref{fig:p_succ_2a}, is partially nullified by the increased annuli area when the equal-area-based allocation is used. Overall, for all three access mechanisms, the equal-interval-based allocation shows a more fair distribution of the SIR success probability between EDs in different annuli.
Fig.~\ref{fig:p_cov} compares the coverage probability of NP-CSMA with the other access schemes for an equal-interval-based allocation. Given that \mbox{S-ALOHA} outperforms both \mbox{NP-CSMA} and \mbox{P-ALOHA} for large SFs, the coverage probability offered by  \mbox{S-ALOHA} is the highest. 

\subsection{Energy Efficiency}
As mention in Sec.~\ref{subsect:ee}, in comparing the energy efficiency of different access mechanisms, we account for the energy consumption contributions of transmissions, synchronization (\mbox{S-ALOHA}) and channel sensing (NP-CSMA). 
The energy efficiency analysis is based on the parameters in Table~\ref{tb:EE_par}.
\begin{table} [!b]
\caption{Energy Model Parameters~\cite{SX1261}}
\centering
\scalebox{0.9}{
\begin{tabular}{ | l | l | l | }
    \hline 
    \textbf{Parameter} & \textbf{Symbol} & \textbf{Value} \\\hline\hline
    Power Consumption in TX  & $P_\text{TX}$ & $84.15$ mW \\
    Power Consumption in RX/CAD  & $P_\text{RX}$ & $15.18$ mW \\ \hline\hline
    Beacon Preamble + Payload& & 10 Symbols + 16 Bytes \\     
    Synchronization Interval   & $T_\text{SYNCH}$ & 128 s\\\hline\hline
    CAD Duration  & $T_\text{CAD}$ &2 symbols \\
    \hline
\end{tabular}}
\label{tb:EE_par}
\end{table}
\begin{figure}[!ht]
     \subfloat[][\centering Energy efficiency, $\bar{N}=3000$\label{subfig:EEa}]{%
       \includegraphics[width=0.486\columnwidth]{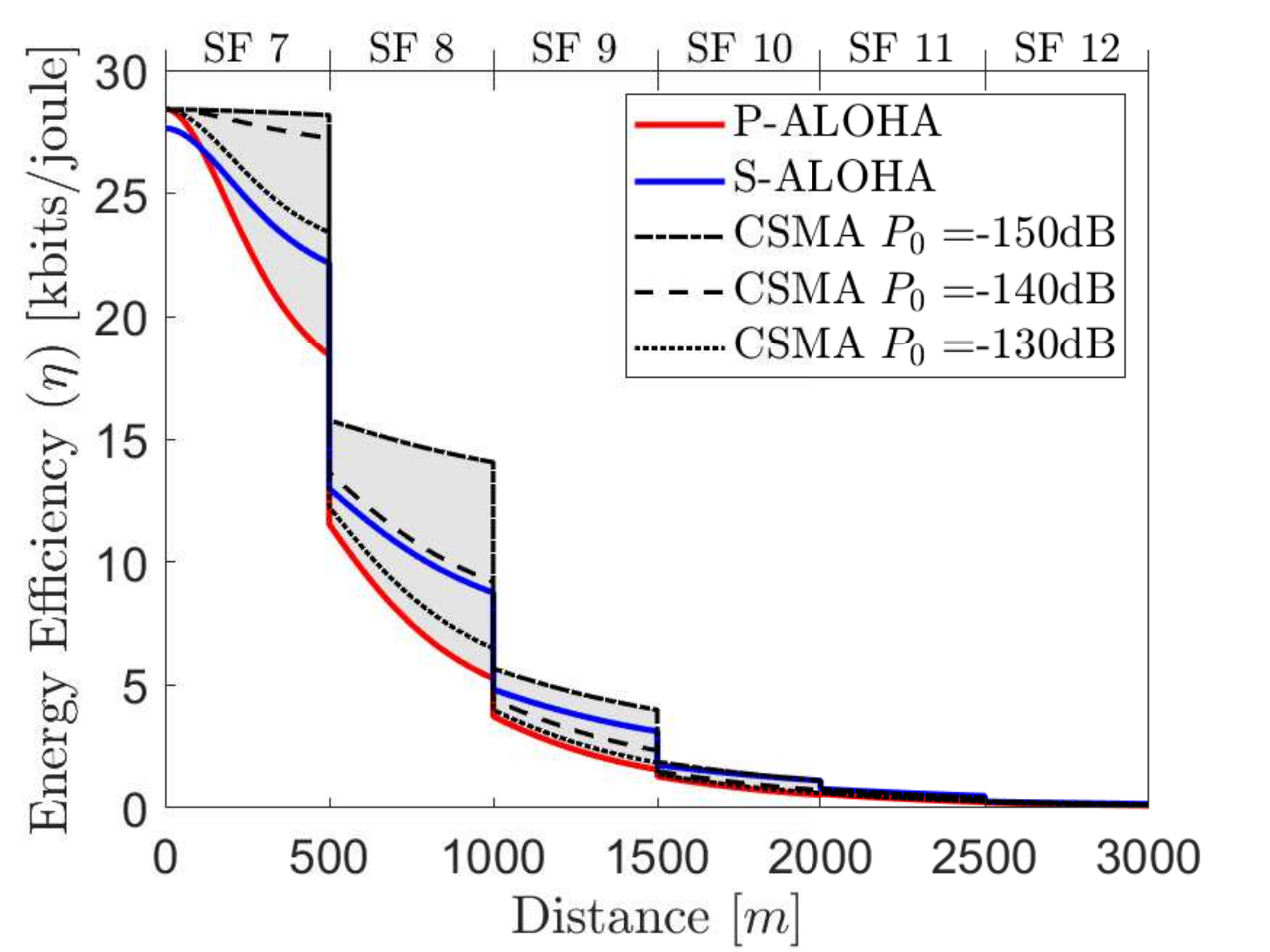}
     }
     \hfill
     \subfloat[][\centering Mean energy efficiency
     
     \label{subfig:EEc}]{%
       \includegraphics[width=0.486\columnwidth]{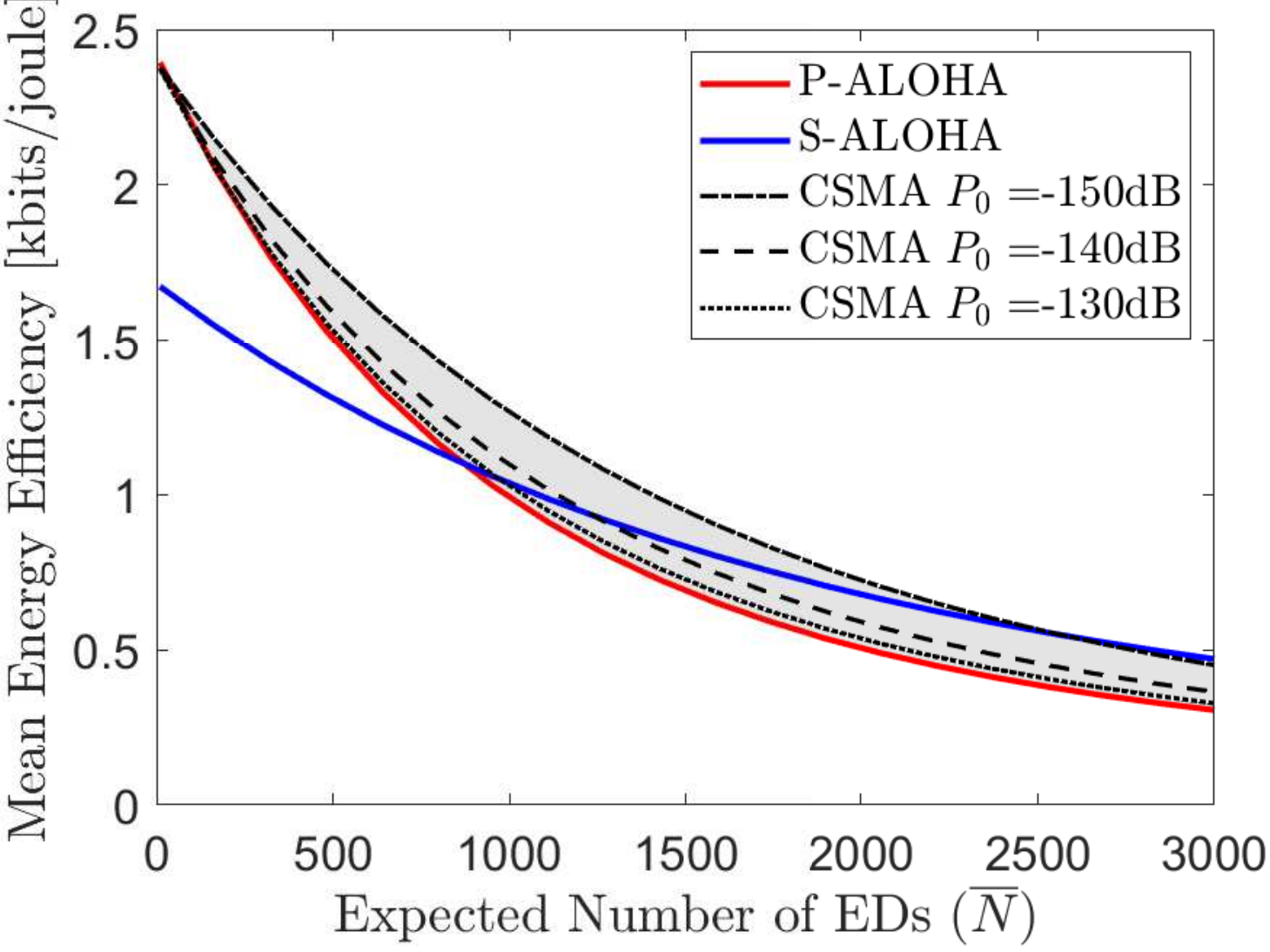}
     }
     \caption{Energy efficiency of \mbox{P-ALOHA}, \mbox{S-ALOHA} and \mbox{NP-CSMA}, \mbox{$R=3$~km}.}
     \vspace{-15pt}
     \label{fig:EE}
\end{figure}
Fig.~\ref{subfig:EEa} shows the energy efficiency of an ED at different distances from the gateway based on \eqref{eq:eta}.
\mbox{S-ALOHA} has a higher energy efficiency compared to \mbox{P-ALOHA}, especially for the EDs located closer to the outer radius of their annulus. For the EDs closest to the gateway, where the success probability of \mbox{P-ALOHA} and \mbox{S-ALOHA} are similar (cf. Fig.~\ref{fig:p_succ_1}), \mbox{P-ALOHA} has a higher energy efficiency than \mbox{S-ALOHA}. In CSMA the energy efficiency depends on the sensing threshold $P_0$, in particular, for small value of $P_0$, the energy efficiency of CSMA is higher than both \mbox{P-ALOHA} and \mbox{S-ALOHA}, remaining almost constant within each annulus.

Fig.~\ref{subfig:EEc} shows the mean energy efficiency for all EDs. At a low density of EDs, \mbox{P-ALOHA} and NP-CSMA offer the best energy efficiency. The energy efficiency of \mbox{S-ALOHA} is reduced by the energy required to maintain the synchronization. When the number of EDs increases, the higher success probability of \mbox{S-ALOHA} more than compensate for the energy used by the synchronization mechanism, making \mbox{S-ALOHA} the most energy-efficient.

To highlight how the proposed model can be used to select the optimal (maximizing energy efficiency) medium access mechanisms, Fig.~\ref{fig:EE_design} shows the operating region of SF~7, 10 and 12 for the parameter space (number of EDs) $\times$ (R). It can be observed that at small R where the hidden terminals are less or at low device density where the number of interferers is small, \mbox{NP-CSMA} surpasses S-ALOHA in energy efficiency. There are regions of the parameter space in which EDs with smaller SFs and hence with a smaller average distance should use \mbox{S-ALOHA}, whereas EDs with larger SFs should use CSMA. This counter-intuitive result is because the energy consumption of the synchronization mechanism via periodic beacons increases exponentially with the SF, making \mbox{S-ALOHA} less energy efficient for large SFs.
\begin{figure}[!t]
\centering
       \includegraphics[width=0.95\columnwidth]{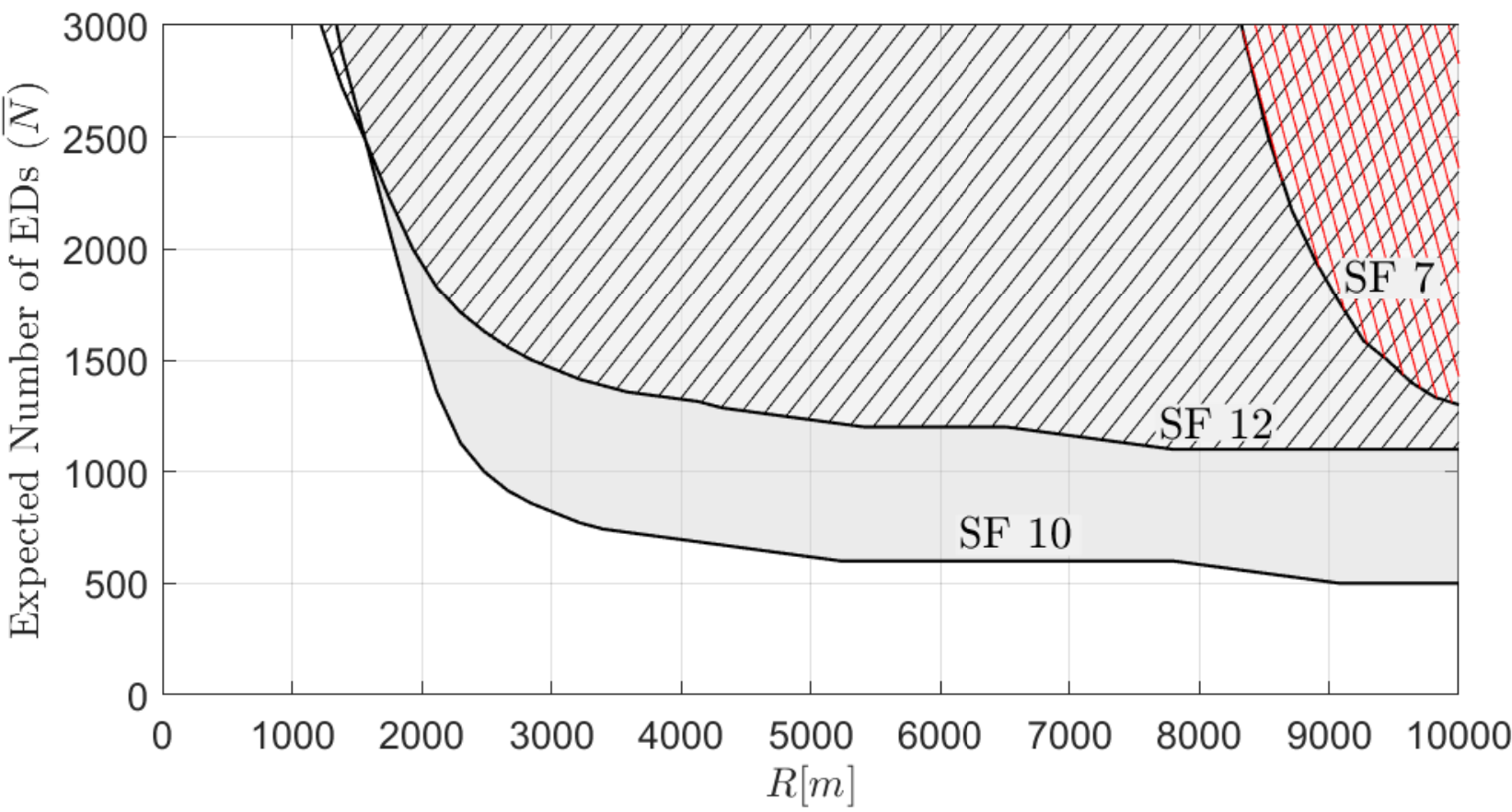}
     \caption{CSMA or S-ALOHA? Optimal operating region in parameter space (SF, $R$ and $\bar{N}$) to maximize the energy efficiency. The area with unique shaded/pattern fill suggests the use of S-ALOHA for an SF, and \mbox{NP-CSMA} otherwise. Regions obtained at $P_0=-140$ dB.}
     \label{fig:EE_design}
\end{figure}

\subsection{Channel Throughput}
The channel throughput obtained from \eqref{eq:Th} for SF~8  and 12 is shown in Fig.~\ref{fig:Tha}  for $R=3$ km.
The small SIR threshold required for the power-capture to take place means that the channel access efficiency of \mbox{P-ALOHA} and \mbox{S-ALOHA} remains high even for a large number of EDs whereas it decreases more rapidly for the theoretical expressions derived without time-power capture effects ($Ge^{-2G}$ and $Ge^{-G}$, respectively).  When multiple messages overlap in time, it is likely for the strongest of these messages to survive the collision. At this point, it is necessary to recall from Sec.~\ref{sec:int_model} that the model uses the dominant interference assuming that no two interferers transmit the same symbol. 
\begin{figure}[!htp]
\centering
       \includegraphics[width=0.9\columnwidth]{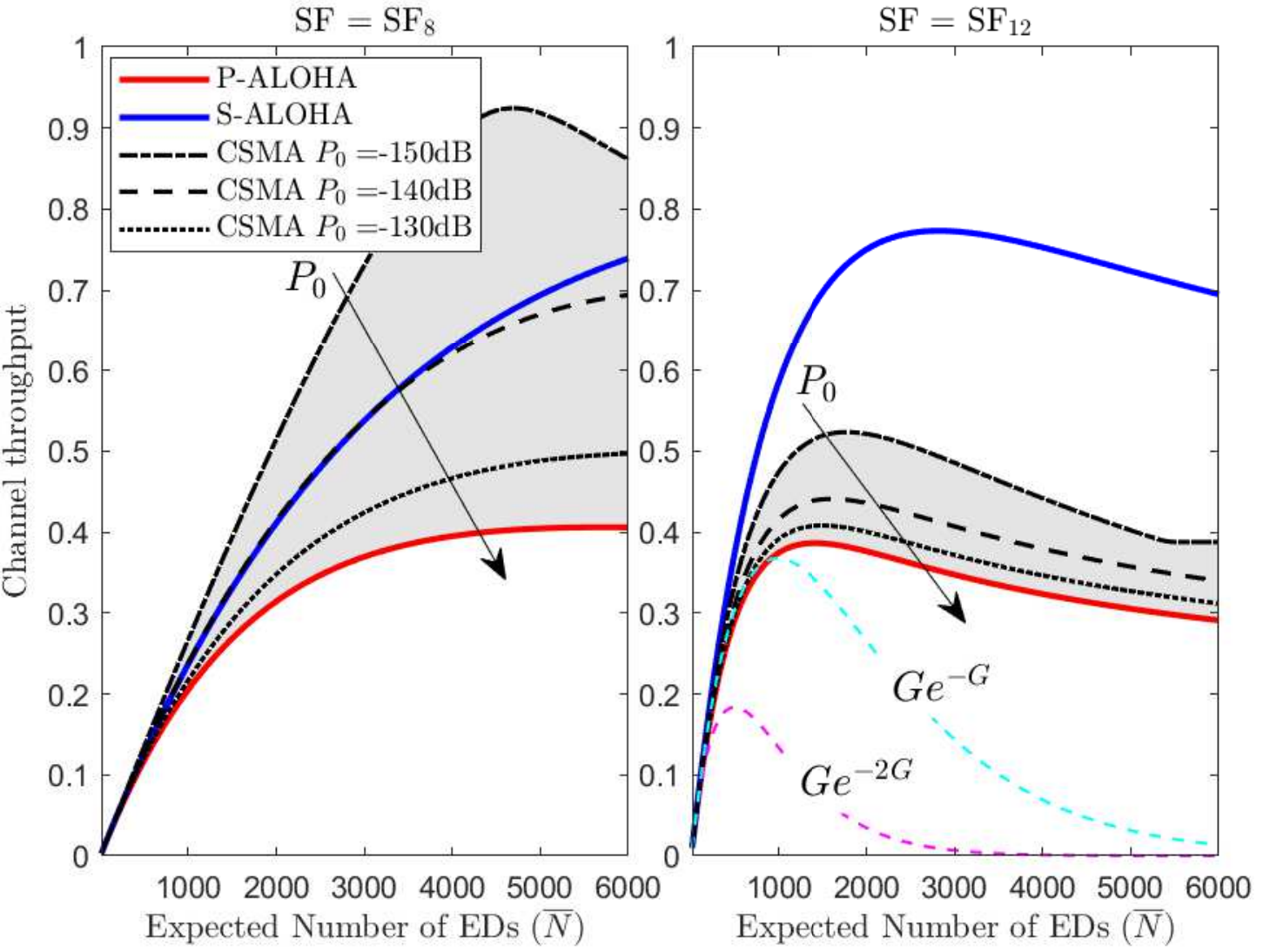}
     \caption{Channel throughput of the studied access mechanisms (the curves for $Ge^{-2G}$ and $Ge^{-G}$ are without capture effect), $R=3$ km}
     \vspace{-5pt}
     \label{fig:Tha}
\end{figure}
\section{Conclusions}
\label{sec:conclusions} 
For a LoRa network, we presented an analytical model and derived metrics for a unified and fair comparison of potential medium access choices (pure ALOHA, slotted ALOHA, and non-persistent CSMA). The model is extensive in the sense that it includes a) an interference model derived from recent experimental and numerical observations in the literature on time- and power-capture offered by LoRa, b) the guard interval and synchronization error for slotted ALOHA, and c) the interference intensity for a selected sensing threshold in CSMA. Our results revealed several assisting guidelines on the design and selection of a medium access solution within LoRa's parameter space: device density, service area, and spreading factor.
The highlight from the results are:
i)	the power-capture effect exhibited by LoRa significantly improves the channel throughput of all the access mechanisms at large device densities.
ii)	the possibility for LoRa messages to survive partial collisions of a message reduces the message vulnerability time and the required guard interval in slotted ALOHA.
iii) slotted ALOHA offers higher reliability than pure ALOHA, at the expense of energy efficiency if the number of devices is small. 
iv)	CSMA can provide better reliability and energy efficiency than slotted ALOHA, but only if the devices are closely located, as is typically the case for small SFs.
In the future, we plan to investigate the delay performance of the different access mechanisms to analyze their suitability in the presence of retransmissions for delay-sensitive industrial applications.

\label{section_conslusions}
\appendix
\subsection{Neighbors in an interfering ED contention domain}
\label{sec:Appendix}
The expected number of active neighbors of a transmitting ED is calculated by using the distance distribution of the EDs. However, a similar approach cannot be used for the interfering node, as the distance distribution of the EDs conditioned on having a transmission in the annulus becomes too complex. The proposed approximation for $E\left[n_A^*\right]$ is presented hereafter.
Let $H$ be the fraction of EDs in an annulus within the contention domain of a transmitter. The value of $H$ can be immediately found from (\ref{eq:H}). 

When $H\approx 0$ (Fig.~\ref{subfig:appendixa}), the EDs outside of the transmitter contention domain are the majority. In this case, the distance distribution of the ED in the annulus before and after removing the transmitter neighbor will be similar and we can approximate $E\left[n_A^*\right]\approx E\left[n_A\right]$.
For $H\approx 1$ (Fig.~\ref{subfig:appendixb}), the EDs outside of the transmitter contention domain are a small fraction of the total number of EDs in the annulus. In this case, we can assume that all the EDs outside of the transmitter contention domain are in the contention domain of a single interfering device,  $E\left[n_A^*\right]\approx p\pi \lambda \left(r_2^2-r_1^2\right)- E\left[n_A\right]$.
For $H\approx 0.5$ (Fig.~\ref{subfig:appendixc}), the EDs outside and inside of the transmitter contention domain are in similar number. In this case, we cannot assume that all the EDs outside of the transmitter contention domain are part of the contention domain of a single interfering device as multiple interfering EDs can be active.  From Fig. \ref{subfig:appendixc}, we can observe that the number of interfering EDs will be two on average and we can approximate  $E\left[n_A^*\right]\approx \frac{1}{2}\left[p\pi \lambda \left(r_2^2-r_1^2\right)- E\left[n_A\right]\right]$.
\begin{figure}[!t]
\vspace{-15pt}
     \subfloat[$H\ll 0.5$ \label{subfig:appendixa}]{%
       \includegraphics[width=0.3\columnwidth]{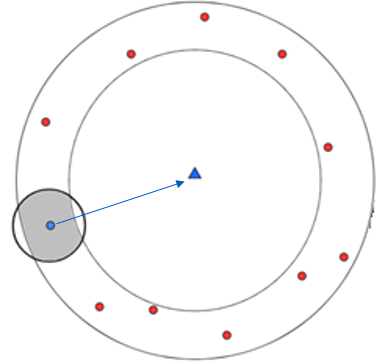}
     }
     \hfill
     \subfloat[$H\approx 0.5$\label{subfig:appendixc}]{%
       \includegraphics[width=0.35\columnwidth]{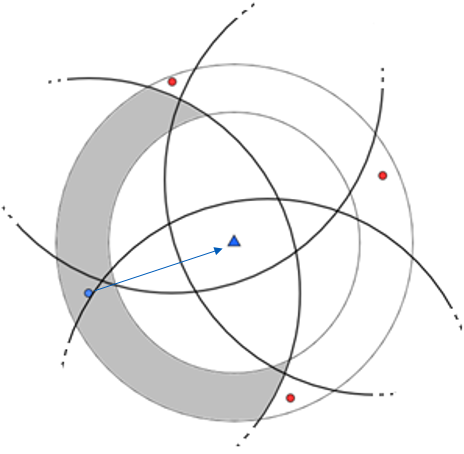}
     }
     \hfill
     \subfloat[$H\gg 0.5$\label{subfig:appendixb}]{%
       \includegraphics[width=0.29\columnwidth]{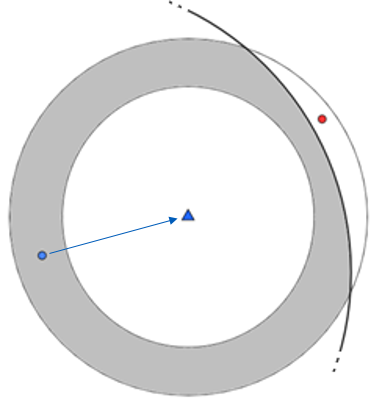}
     }
     \caption{Representation of the fraction of  EDs located  inside  the  transmitter contention domain}
     \vspace{-15pt}
     \label{fig:appendix}
\end{figure}

We introduce an exponential function $g(H)$ (Fig. \ref{subfig:g_fun}) such that $g(H)=1$ for $H \in \{0,1\}$ and $g(H)=2$ for $H=0.5$ and we use it to scale the upper bound to  $E\left[n_A^*\right]$ found in \eqref{eq:bound}.
\begin{equation}
\displaystyle
\label{eq:g(H)}
	g(H)=1+\exp\left(1-\frac{1}{4H-4H^2}\right)\text{.}
\end{equation}
Fig. \ref{subfig:g_fun_val} compares the number of interfering EDs obtained by simulations and the proposed mathematical model.
\begin{figure}[!t]
     \subfloat[\label{subfig:g_fun}]{%
       \includegraphics[width=0.48\columnwidth]{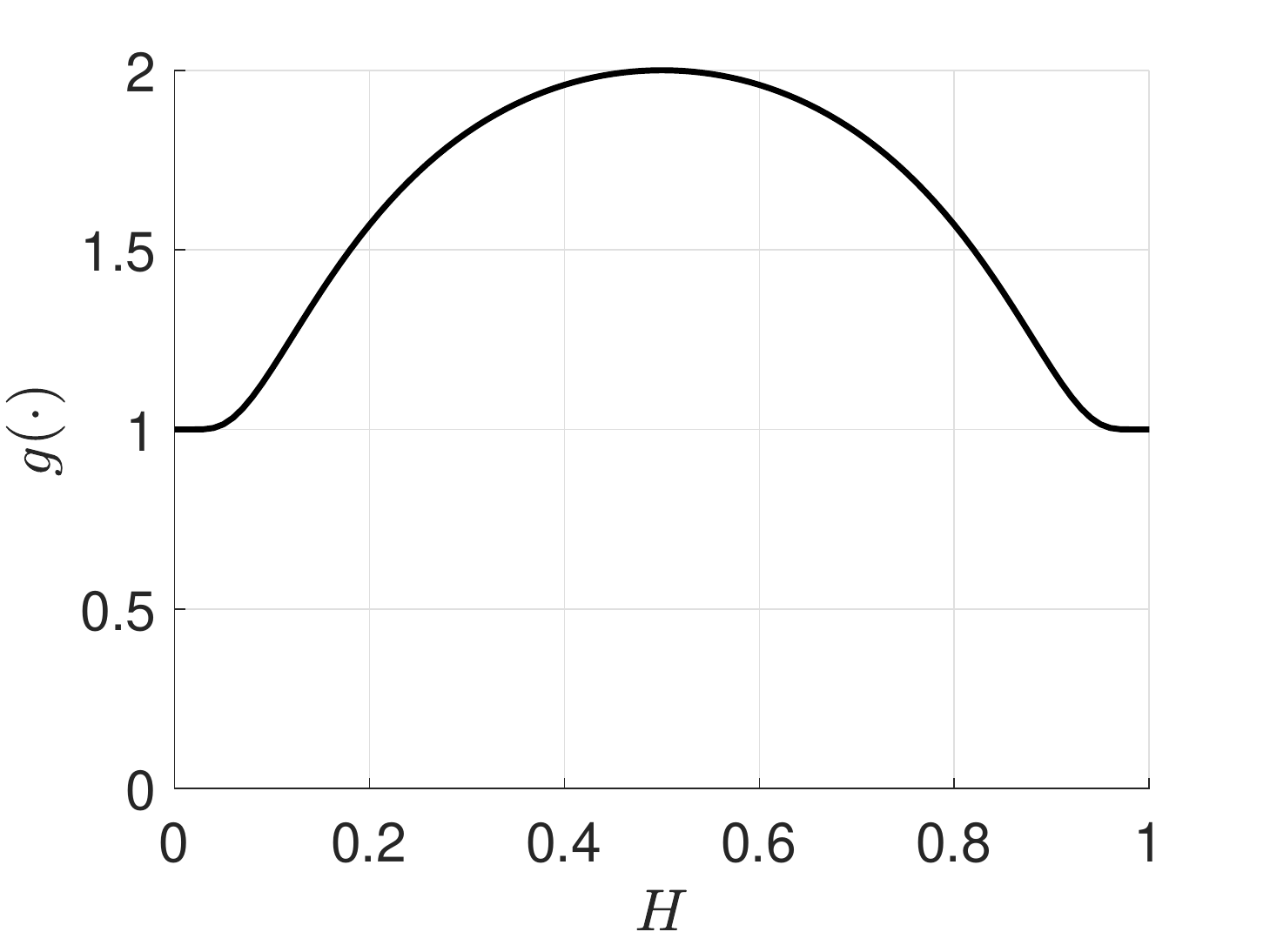}
     }
     \hfill
     \subfloat[\label{subfig:g_fun_val}]{%
       \includegraphics[width=0.48\columnwidth]{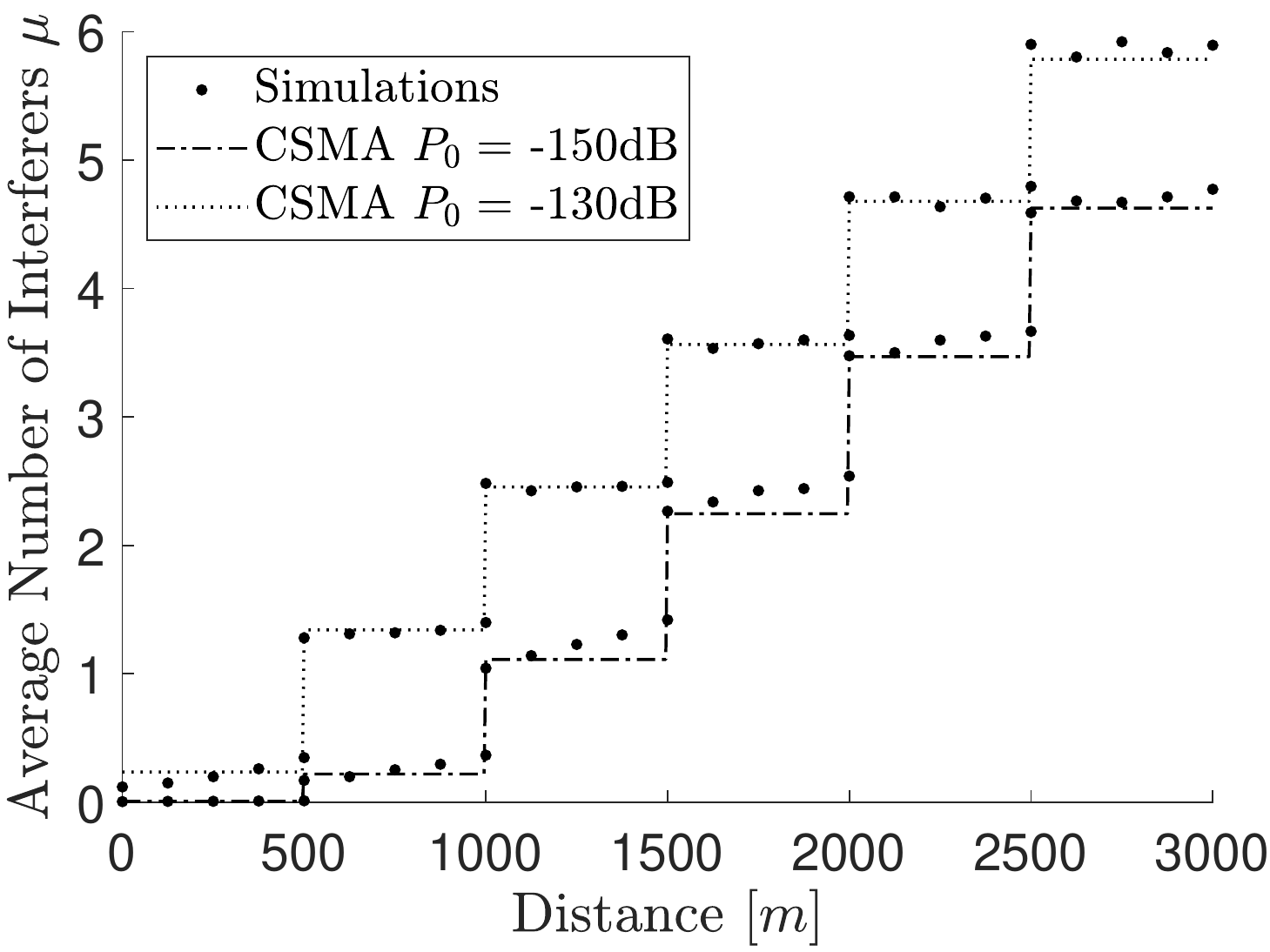}
     }
     \caption{Scaling function for the interfering ED and number of interferers}
     \vspace{-15pt}
     \label{fig:appendix2}

\end{figure}

\bibliographystyle{IEEEtran}
\bibliography{refs}

\begin{thebibliography}{10}
\providecommand{\url}[1]{#1}
\csname url@samestyle\endcsname
\providecommand{\newblock}{\relax}
\providecommand{\bibinfo}[2]{#2}
\providecommand{\BIBentrySTDinterwordspacing}{\spaceskip=0pt\relax}
\providecommand{\BIBentryALTinterwordstretchfactor}{4}
\providecommand{\BIBentryALTinterwordspacing}{\spaceskip=\fontdimen2\font plus
\BIBentryALTinterwordstretchfactor\fontdimen3\font minus
  \fontdimen4\font\relax}
\providecommand{\BIBforeignlanguage}[2]{{%
\expandafter\ifx\csname l@#1\endcsname\relax
\typeout{** WARNING: IEEEtran.bst: No hyphenation pattern has been}%
\typeout{** loaded for the language `#1'. Using the pattern for}%
\typeout{** the default language instead.}%
\else
\language=\csname l@#1\endcsname
\fi
#2}}
\providecommand{\BIBdecl}{\relax}
\BIBdecl

\bibitem{Candell}
R.~{Candell}, M.~{Kashef}, Y.~{Liu}, K.~B. {Lee}, and S.~{Foufou}, ``Industrial
  wireless systems guidelines: Practical considerations and deployment life
  cycle,'' \emph{IEEE Ind. Electron. Mag.}, vol.~12, no.~4, pp. 6--17, Dec
  2018.

\bibitem{sanchez2016state}
R.~Sanchez-Iborra and M.-D. Cano, ``{State of the art in LP-WAN solutions for
  industrial IoT services},'' \emph{Sensors}, vol.~16, no.~5, p. 708, 2016.

\bibitem{luvisotto2018use}
M.~Luvisotto, F.~Tramarin, L.~Vangelista, and S.~Vitturi, ``{On the use of
  LoRaWAN for indoor industrial IoT applications},'' \emph{Wireless
  Communications and Mobile Computing}, 2018.

\bibitem{centenaro2016long}
M.~Centenaro, L.~Vangelista, A.~Zanella, and M.~Zorzi, ``Long-range
  communications in unlicensed bands: The rising stars in the {IoT} and smart
  city scenarios,'' \emph{IEEE Wireless Commun.}, vol.~23, no.~5, pp. 60--67,
  2016.

\bibitem{georgiou2017low}
O.~Georgiou and U.~Raza, ``{Low power wide area network analysis: Can {LoRa}
  scale?}'' \emph{IEEE Wireless Commun. Lett.}, vol.~6, no.~2, pp. 162--165,
  2017.

\bibitem{mahmood2019scalability}
A.~Mahmood, E.~Sisinni, L.~Guntupalli, R.~Rond{\'o}n, S.~A. Hassan, and
  M.~Gidlund, ``{Scalability analysis of a LoRa network under imperfect
  orthogonality},'' \emph{IEEE Trans. Ind. Informat.}, vol.~15, no.~3, pp.
  1425--1436, 2019.

\bibitem{liando2019known}
J.~C. Liando, A.~Gamage, A.~W. Tengourtius, and M.~Li, ``{Known and unknown
  facts of LoRa: Experiences from a large-scale measurement study},'' \emph{ACM
  TOSN}, vol.~15, no.~2, p.~16, 2019.

\bibitem{croce2018impact}
D.~Croce, M.~Gucciardo, S.~Mangione, G.~Santaromita, and I.~Tinnirello,
  ``{Impact of LoRa imperfect orthogonality: Analysis of link-level
  performance},'' \emph{IEEE Commun. Lett.}, vol.~22, no.~4, pp. 796--799,
  2018.

\bibitem{rahmadhani2018lorawan}
A.~Rahmadhani and F.~Kuipers, ``{When LoRaWAN Frames Collide},'' in \emph{ACM
  WiNTECH}, 2018.

\bibitem{haxhibeqiri2017lora}
J.~Haxhibeqiri, F.~Van~den Abeele, I.~Moerman, and J.~Hoebeke, ``{Lora
  scalability: A simulation model based on interference measurements},''
  \emph{Sensors}, vol.~17, no.~6, p. 1193, 2017.

\bibitem{polonelli2019slotted}
T.~Polonelli, D.~Brunelli, A.~Marzocchi, and L.~Benini, ``{Slotted ALOHA on
  LoRaWAN-design, analysis, and deployment},'' \emph{Sensors}, vol.~19, no.~4,
  p. 838, 2019.

\bibitem{pham2018robust}
C.~Pham, ``{Robust CSMA for long-range LoRa transmissions with image sensing
  devices},'' in \emph{Wireless Days (WD)}, 2018, pp. 116--122.

\bibitem{to2018simulation}
T.-H. To and A.~Duda, ``{Simulation of LoRa in ns-3: Improving LoRa performance
  with CSMA},'' in \emph{IEEE ICC}, 2018, pp. 1--7.

\bibitem{leonardi2018industrial}
L.~Leonardi, F.~Battaglia, G.~Patti, and L.~L. Bello, ``{Industrial LoRa: A
  novel medium access strategy for LoRa in industry 4.0 applications},'' in
  \emph{IEEE IECON}, 2018, pp. 4141--4146.

\bibitem{piyare2018demand}
R.~Piyare, A.~Murphy, M.~Magno, and L.~Benini, ``{On-Demand LoRa: Asynchronous
  TDMA for Energy Efficient and Low Latency Communication in IoT},''
  \emph{Sensors}, vol.~18, no.~11, p. 3718, 2018.

\bibitem{haxhibeqiri2018low}
J.~Haxhibeqiri, I.~Moerman, and J.~Hoebeke, ``Low overhead scheduling of {LoRa}
  transmissions for improved scalability,'' \emph{IEEE Internet Things J.},
  vol.~6, no.~2, pp. 3097--3109, 2018.

\bibitem{zorbas2020ts}
D.~Zorbas, K.~Abdelfadeel, P.~Kotzanikolaou, and D.~Pesch, ``{TS-LoRa:
  Time-slotted LoRaWAN for the Industrial Internet of Things},'' \emph{Computer
  Communications}, 2020.

\bibitem{rizzi2017using}
M.~Rizzi, P.~Ferrari, A.~Flammini, E.~Sisinni, and M.~Gidlund, ``{Using LoRa
  for industrial wireless networks},'' in \emph{IEEE WFCS}, 2017, pp. 1--4.

\bibitem{haenggi2009interference}
M.~Haenggi \emph{et~al.}, ``Interference in large wireless networks,''
  \emph{Foundations and Trends{\textregistered} in Networking}, vol.~3, no.~2,
  pp. 127--248, 2009.

\bibitem{baccelli2010stochastic}
F.~Baccelli, B.~B{\l}aszczyszyn \emph{et~al.}, ``{Stochastic geometry and
  wireless networks: Volume II Applications},'' \emph{Foundations and
  Trends{\textregistered} in Networking}, vol.~4, no. 1--2, pp. 69--78, 2010.

\bibitem{haxhibeqiri2017lora2}
J.~Haxhibeqiri, A.~Karaagac, F.~Van~den Abeele, W.~Joseph, I.~Moerman, and
  J.~Hoebeke, ``{LoRa indoor coverage and performance in an industrial
  environment: Case study},'' in \emph{IEEE ETFA}, 2017, pp. 1--8.

\bibitem{sorensen2019analysis}
R.~B. S{\o}rensen, N.~Razmi, J.~J. Nielsen, and P.~Popovski, ``{Analysis of
  LoRaWAN uplink with multiple demodulating paths and capture effect},''
  \emph{arXiv preprint arXiv:1902.02866}, 2019.

\bibitem{crozier1990sloppy}
S.~N. Crozier, ``{Sloppy-slotted ALOHA},'' in \emph{{Proc. of the Second
  International Mobile Satellite Conference}}, 1990, pp. 357--362.

\bibitem{mathai1999introduction}
A.~M. Mathai, \emph{{An introduction to geometrical probability: distributional
  aspects with applications}}.\hskip 1em plus 0.5em minus 0.4em\relax CRC
  Press, 1999, vol.~1.

\bibitem{CAD}
{Semtech}, \emph{{Application Note: SX126x CAD Performance Evaluation}},
  {AN1200.48, Rev 1.0}, November 2018.

\bibitem{SX1261}
\emph{SX1261/2 Long Range, Low Power, sub-GHz RF Transceiver}, Semtech, 12
  2017, rev. 1.1.

\end{thebibliography}

\end{document}